%% file: ExpurgatedPaperFull.tex
\documentclass[english,draftcls, onecolumn]{IEEEtran}
\usepackage[T1]{fontenc}
\usepackage[latin9]{inputenc}
\usepackage{geometry}
\usepackage{amsthm}
\usepackage{amsmath}
\usepackage{amssymb}
\usepackage{cite}
\usepackage{graphicx}
\usepackage{epstopdf}
\usepackage{float} 
 
\makeatletter

\theoremstyle{plain}
\newtheorem{thm}{\protect\theoremname}
\theoremstyle{plain}
\newtheorem{lem}{\protect\lemmaname}
  \theoremstyle{plain} 
  \newtheorem{prop}{\protect\propositionname}
 
\input{preamble.tex} 

\DeclareMathOperator*{\argmax}{arg\,max}
\DeclareMathOperator*{\argmin}{arg\,min}

\newcommand{\Ic}{\mathcal{I}}
\newcommand{\hover}{\overline{h}}

\makeatother 
\usepackage{babel}
\providecommand{\theoremname}{Theorem}
\providecommand{\lemmaname}{Lemma}
\providecommand{\propositionname}{Proposition}

\begin{document}
\title{Expurgated Random-Coding Ensembles: Exponents, Refinements and Connections}
\author{Jonathan Scarlett, Li Peng, Neri Merhav, Alfonso Martinez and Albert Guill\'en i F\`abregas}
\maketitle

\begin{abstract} 
    This paper studies expurgated random-coding bounds and exponents for channel coding with a given 
    (possibly suboptimal) decoding rule.  Variations of Gallager's analysis
    are presented, yielding several asymptotic and non-asymptotic bounds 
    on the error probability for an arbitrary codeword distribution.
    A simple non-asymptotic bound is shown to attain an exponent
    of Csisz\'{a}r and K\"{o}rner under constant-composition coding.
    Using Lagrange duality, this exponent is expressed in several forms,
    one of which is shown to permit a direct derivation via cost-constrained
    coding which extends to infinite and continuous alphabets.
    The method of type class enumeration is studied, and it is shown that this 
    approach can yield improved exponents and better tightness guarantees 
    for some codeword distributions.  A generalization of this approach
    is shown to provide a multi-letter exponent which extends immediately 
    to channels with memory.  Finally, a refined analysis expurgated i.i.d. random coding
    is shown to yield a $O\big(\frac{1}{\sqrt n}\big)$ prefactor,
    thus improving on the standard $O(1)$ prefactor.  Moreover, the implied
    constant is explicitly characterized.
\end{abstract}

\symbolfootnote[0]{
J. Scarlett and L. Peng are with the Department of Engineering, University of Cambridge, Cambridge, CB2 1PZ, U.K. (e-mails: jmscarlett@gmail.com, lp327@cam.ac.uk).  
N.~Merhav is with the Department of Electrical Engineering, Technion -  Israel Institute of Technology, Technion City, Haifa 32000, Israel. (e-mail: merhav@ee.technion.ac.il).
A. Martinez is with the Department of Information and Communication Technologies, Universitat Pompeu Fabra, 08018 Barcelona, 
Spain (e-mail: alfonso.martinez@ieee.org).  A. Guill\'en i F\`abregas is with the Instituci\'o Catalana de Recerca i Estudis Avan\c{c}ats (ICREA), 
the Department of Information and Communication Technologies, Universitat Pompeu Fabra, 08018 Barcelona, Spain, 
and also with the Department of Engineering, University of Cambridge, Cambridge, CB2 1PZ, U.K. (e-mail: guillen@ieee.org).

This work has been funded in part by the European Research Council under ERC grant agreement 259663, by the European Union's 7th Framework Programme (PEOPLE-2011-CIG) 
under grant agreement 303633 and by the Spanish Ministry of Economy and Competitiveness under grants RYC-2011-08150 and TEC2012-38800-C03-03.
The work of N.~Merhav was partially supported by the Israel Science Foundation (ISF), grant no.\ 412/12. 

This is the extended version of a paper which was accepted to \emph{IEEE Transactions on Information Theory}
(April 2014).  A shortened version was presented at the 2014 International Zurich Seminar
on Communications.} 

\thispagestyle{empty}

\vspace{-1cm}

\section{Introduction} \label{sec:EXP_INTRODUCTION}

Achievable performance bounds for channel coding are typically obtained by analyzing the average error probability of 
an ensemble of codebooks with independently generated codewords.  For memoryless channels, random  
codes with independent and identically distributed (i.i.d.) symbols achieve the channel capacity \cite{Shannon}, 
characterize the error exponent of the best code at sufficiently high rates \cite[Ch. 5]{Gallager}, 
and provide tight bounds on the finite-length performance \cite{Finite}.

At low rates, the error probability of the best code in the random-coding ensemble can be significantly smaller 
than the average.  In such cases, better performance bounds are obtained by considering an ensemble in which a
subset of the randomly generated codewords are expurgated from the codebook.  In particular,
the error exponents resulting from such techniques are generally higher than the random-coding error exponent at low
rates.  Existing works exploring such techniques include those of Gallager \cite[Sec. 5.7]{Gallager},
Csisz\'{a}r-K\"{o}rner-Marton \cite{ExpurgCKM},\cite[Ex. 10.18]{CsiszarBook} and Csisz\'{a}r-K\"{o}rner \cite{Csiszar1}.  
The advantages of Gallager's approach include its simplicity and the fact that the analysis 
is not restricted to finite alphabets.  On the other hand, as we will see in 
Section \ref{sec:EXP_OVERVIEW}, the exponents of \cite{Csiszar1,ExpurgCKM,CsiszarBook}  can improve  
on that of Gallager for a given input distribution or decoding rule. 

In this paper, we provide techniques that attain the best of each of the above approaches.
Using variations of Gallager's analysis, we obtain several asymptotic and
non-asymptotic bounds for an arbitrary codeword distribution.
Using these bounds, we provide derivations of both new and existing expurgated exponents, each 
yielding various advantages such as simplicity, generality, and guarantees of exponential tightness.  
We explore the method of type class enumeration 
(e.g. see \cite{MerhavErasure,MerhavIC,MerhavPhysics}) for both discrete and continuous channels, 
and show that it can yield improved exponents and tightness guarantees, as well as providing
a multi-letter exponent which extends immediately to channels with memory. 

\subsection{System Setup} \label{sec:EXP_SYSTEM_SETUP}

The input and output alphabets are denoted by $\Xc$ and $\Yc$
respectively.  The channel is assumed to be memoryless, yielding
an $n$-letter transition law given by $W^{n}(\yv|\xv)\defeq\prod_{i=1}^{n}W(y_{i}|x_{i})$
for some conditional distribution $W(y|x)$. 
In the case that both $\Xc$ and $\Yc$ are finite, the channel is a 
discrete memoryless channel (DMC), but we do not assume this to be the case in general.
The encoder takes as input a message $m$ equiprobable on the 
set $\{1,\dotsc,M\}$, and transmits the corresponding codeword $\xv^{(m)}$ 
from a codebook $\Cc=\{\xv^{(1)},\dotsc,\xv^{(M)}\}$.
The decoder receives the vector $\yv$ at the output of
the channel, and forms the estimate 
    \begin{equation}
    \hat{m} = \argmax_{j\in\{1,\dotsc,M\}}q^{n}(\xv^{(j)},\yv), \label{eq:SU_DecodingRule}
    \end{equation}
where $q^{n}(\xv,\yv)\defeq\prod_{i=1}^{n}q(x_{i},y_{i})$, and $q(x,y)$ is
a non-negative function called the \emph{decoding metric}. An error is said to have occurred
if $\hat{m} \ne m$, and we assume that ties are broken as errors.  We let $p_{e,m}(\Cc)$ be the 
error probability induced by $\Cc$ given a particular message $m$, and we denote 
the maximal error probability by $p_{e}(\Cc) \triangleq \max_{m}p_{e,m}(\Cc)$.

When $q(x,y)=W(y|x)$, \eqref{eq:SU_DecodingRule} is the optimal maximum-likelihood (ML) decoding rule.
For other decoding metrics, this setting is that of \emph{mismatched decoding} 
\cite{Merhav,Csiszar2,MMRevisited,JournalSU}, which is
relevant when ML decoding is not feasible, e.g. due to channel uncertainty or
implementation constraints.

Throughout the paper, we consider channels with both constrained and unconstrained inputs.
In the former setting, each codeword $\xv$ must satisfy a constraint of the form  
    \begin{equation}
        \frac{1}{n}\sum_{i=1}^{n} c(x_{i}) \le \Gamma, \label{eq:EXP_SystemCost}
    \end{equation}
where $c(\cdot)$ is referred to as a cost function, and $\Gamma$ is a constant.
Except where stated otherwise, it will be assumed that the input is unconstrained, which
corresponds to $\Gamma=\infty$.

For a given rate $R$, an error exponent $E(R)$ is said to be achievable if  
there exists a sequence of codebooks $\Cc_{n}$ of length $n$ and rate $R$ whose error
probability $p_{e}(\Cc_{n})$ satisfies
    \begin{equation} 
    \liminf_{n\to\infty}-\frac{1}{n}\log p_{e}(\Cc_{n})\ge E(R).
    \end{equation} 
We focus on the maximal error probability rather than the average error probability, 
but the two are equivalent for the purposes of studying error exponents.

\subsection{Previous Work} \label{sec:EXP_PREVIOUS_WORK}

Considering ML decoding, Gallager \cite[Ch. 5]{Gallager} studied an ensemble
in which $2M-1$ codewords are generated at random, and a subset of $M$ codewords forms the codebook.  Roughly speaking, the
codewords which are kept are those which have the lowest error probability among the original codewords.  A 
different approach was taken by Csisz\'{a}r, K\"{o}rner and Marton \cite{ExpurgCKM} (see also \cite[Ex. 10.18]{CsiszarBook}), 
who began by proving the existence of a collection of constant-composition codewords such that any two codewords have
a joint empirical distribution satisfying certain properties.  By analyzing this collection of codewords using the  
method of types, an error exponent was obtained which coincides with that of Gallager after the optimization of
the input distribution.  An exponent for mismatched decoding was derived by Csisz\'{a}r and K\"{o}rner \cite{Csiszar1}, 
and was shown to coincide with that of \cite{ExpurgCKM} when particularized to the case of ML decoding.

As stated in the introduction, 
the exponents of \cite{ExpurgCKM,Csiszar1} can in fact improve on that of Gallager 
for a given input distribution. However, the proofs rely heavily on techniques which are
valid only when the input and output alphabets are finite.  In particular, \cite{ExpurgCKM} 
uses the type packing lemma \cite[Ch. 10]{CsiszarBook}, and \cite{Csiszar1} uses a combinatorial
graph decomposition lemma.  For other related works, see 
\cite{ExpurgJelenik, ExpurgBlahut, ExpurgOmura, JSCC4}.

Overviews of the mismatched decoding problem can be found in \cite{Merhav,Csiszar2,MMRevisited,JournalSU}. 
Most of the literature has focused on achievable rates, whereas this paper is concerned with the performance
at low rates.  The mismatched decoding paper most relevant to this one is \cite{JournalSU}, which studies
random-coding error exponents for various non-expurgated ensembles.

\subsection{Contributions} \label{sec:EXP_CONTRIBUTIONS}

The main contributions of this paper are as follows:
\begin{itemize}
  \item In Section \ref{sec:EXP_BOUNDS}, we present variations of Gallager's analysis 
        which yield several asymptotic and non-asymptotic bounds on the error probability.
        In particular, we consider the use of a logarithmic function in the expurgation
        argument in place of the power function used by Gallager \cite[Sec. 7.3]{Gallager}.
  \item In Section \ref{sec:EXP_OVERVIEW}, we present an overview of various expurgated
        exponents and the connections between them.  Using the method of Lagrange duality \cite{Convex}, we relate
        the exponents given in \cite{Gallager,ExpurgCKM,Csiszar1}.  Generalizations of the 
        exponents in \cite{Gallager,ExpurgCKM} to the setting of mismatched decoding are given,
        and an alternative form of the exponent in \cite{Csiszar1} is given which extends
        readily to channels with infinite or continuous alphabets.
  \item In Section \ref{sec:EXP_DISCRETE}, we present several methods for deriving both new and existing
        exponents:
        \begin{itemize}
            \item In Section \ref{sec:EXP_DISC_SIMPLE}, we present simple techniques for deriving
                  exponents using a non-asymptotic bound from Section \ref{sec:EXP_BOUNDS}.
                  Applying constant-composition coding and the method of types recovers
                  the exponent in \cite{Csiszar1}, thus providing a simple and concise proof.
                  Furthermore, applying cost-constrained coding with multiple 
                  auxiliary costs \cite{JournalSU} recovers the generalization of this exponent to more general alphabets. 
            \item In Section \ref{sec:EXP_TYPE_ENUM}, we study the method of type
                  class enumeration (e.g. see \cite{MerhavErasure,MerhavIC,MerhavPhysics}),
                  which is shown  to yield better exponents than the simpler approach 
                  for some codeword distributions, as well as better guarantees of exponential tightness.
            \item In Section \ref{sec:EXP_DIST_ENUM}, we extend the type class enumeration
                  analysis to allow for infinite and continuous alphabets.  This is 
                  not only of interest in itself, but also yields 
                  a multi-letter exponent which can be directly applied to channels with memory
                  and more general decoding metrics.
        \end{itemize}
  \item In Section \ref{sec:EXP_PREFACTOR}, we present a refined derivation of Gallager's exponent
        for i.i.d. random coding (and its generalization to mismatched decoding) with a 
        $O\big(\frac{1}{\sqrt{n}}\big)$ prefactor, thus improving on the original $O(1)$ prefactor.
        Similar improvements for the non-expurgated random-coding error exponent have recently been
        obtained by Altu\u{g} and Wagner \cite{RefinementJournal} (see also \cite{PaperRefinement}).
\end{itemize}

\subsection{Notation} \label{sec:EXP_NOTATION}

We use bold symbols for vectors (e.g. $\xv$), and denote the corresponding
$i$-th entry using a subscript (e.g. $x_{i}$).

The set of all probability distributions on an alphabet, say $\Xc$,
is denoted by $\Pc(\Xc)$, and the set of all empirical
distributions on a vector in $\Xc^{n}$ (i.e. types \cite[Ch. 2]{CsiszarBook})
is denoted by $\Pc_{n}(\Xc)$. For a given type $Q\in\Pc_{n}(\Xc)$, the type class $T^{n}(Q)$
is defined to be the set of all sequences in $\Xc^{n}$ with type $Q$.

The probability of an event is denoted by $\PP[\cdot]$, and
the symbol $\sim$ means ``distributed as''. The marginals of a
joint  distribution $P_{XY}(x,y)$ are denoted by $P_{X}(x)$ and
$P_{Y}(y)$. We write $P_{X}=\Ptilde_{X}$ to denote element-wise
equality between two probability distributions on the same alphabet.
Expectation with respect to a joint distribution $P_{XY}(x,y)$ is
denoted by $\EE_{P}[\cdot]$, or simply $\EE[\cdot]$ when the associated 
probability distribution is understood from the context.
Similarly, the mutual information with respect
to $P_{XY}$ is written as $I_{P}(X;Y)$, or simply $I(X;Y)$ when
the distribution is understood from the context. Given a distribution
$Q(x)$ and conditional distribution $W(y|x)$, we write $Q\times W$
to denote the joint distribution defined by $Q(x)W(y|x)$.

For two positive sequences $f_{n}$ and $g_{n}$, we write $f_{n}\doteq g_{n}$
if $\lim_{n\to\infty}\frac{1}{n}\log\frac{f_{n}}{g_{n}}=0$, and we write $f_{n}\,\dot{\le}\,g_{n}$
if $\limsup_{n\to\infty}\frac{1}{n}\log\frac{f_{n}}{g_{n}}\le0$, and analogously for $\dot{\ge}$.
We write $f_{n}=O(g_{n})$ if $|f_{n}|\le c|g_{n}|$
for some $c$ and sufficiently large $n$. 
All logarithms have base $e$, and all rates are in units of nats
except in the examples, where bits are used. We define $[c]^{+}=\max\{0,c\}$,
and denote the indicator function by $\openone\{\cdot\}$.

\section{Expurgated Bounds} \label{sec:EXP_BOUNDS}

In this section, we present a number of variations of Gallager's bounds and techniques
which will provide the starting points of the derivations of the exponents in Section
\ref{sec:EXP_DISCRETE}.  We let
$P_{\Xv}$ denote a codeword distribution, and we define the random variables
$(\Xv,\Yv,\Xvbar)$ distributed according to
    \begin{equation}
    (\Xv,\Yv,\Xvbar) \sim \PXv(\xv)W^{n}(\yv|\xv)\PXv(\xvbar).
    \end{equation}
In the case that a cost constraint of the form \eqref{eq:EXP_SystemCost} is present,
we assume that $P_{\Xv}$ is chosen such that $\Xv$ satisfies the constraint with probability one. 

We let $\Csf=\{\Xv^{(1)},\dotsc,\Xv^{(M')}\}$ be a random codebook of size $M'$ with each 
codeword independently generated according to $P_{\Xv}$.  The symbol $\Cc$ is used to 
denote a fixed expurgated codebook containing $M \le M'$ codewords.  

We begin with the following straightforward generalization of \cite[Lemma, p. 151]{Gallager}.

\begin{lem} \label{lem:EXP_GalLemma}
    Fix a function $f\,:\,[0,1] \to \RR$ and a codeword distribution $P_{\Xv}$ such
    that $f(p_{e,m}(\Csf))$ is non-negative for all $m$ with probability one.  
    For any $\eta>0$, there exists a codebook $\Cc$ of size $M$ such that
    $M'\frac{\eta}{1+\eta} < M \le M'$ and
    \begin{equation}
        f\big(p_{e,m}(\Cc)\big) \le (1+\eta)\EE\big[f(p_{e,m}(\Csf))\big] \label{eq:EXP_GalLemma}
    \end{equation}
    for $m=1,\dotsc,M$.
\end{lem}  
\begin{proof}
    The proof is identical to \cite[Lemma, p. 151]{Gallager}, with the assumption
    of $f(p_{e,m}(\Csf))$ being non-negative ensuring the validity of Markov's inequality.
\end{proof}

While Lemma \ref{lem:EXP_GalLemma} is valid for any function $f(\cdot)$, it is
primarily of interest when $f(\cdot)$ is monotonically increasing, so that 
\eqref{eq:EXP_GalLemma} can be inverted in order to obtain an upper bound on $p_{e,m}(\Cc)$.
Gallager \cite{Gallager} presented the lemma with the choices
$\eta=1$ and $f(\cdot)=(\cdot)^{1/\rho}$, where $\rho>0$, thus proving the existence of   
a codebook $\Cc$ of size $M$ such that
\begin{equation}
    p_{e}(\Cc) \le \Big(2\EE\big[p_{e,m}(\Csf)^{1/\rho}\big]\Big)^{\rho}, \label{eq:EXP_GalBound}
\end{equation}
where $\Csf$ contains $M'=2M-1$ codewords.
In the following theorem, we provide non-asymptotic bounds on the error 
probability which follow in a straightforward fashion from \eqref{eq:EXP_GalBound}.   
The proof alters Gallager's arguments for the purpose of better characterizing 
the non-asymptotic performance, and also for dealing with suboptimal decoding rules.    

\begin{thm} \label{thm:EXP_Finite}
    For any pair $(n,M)$, codeword distribution $\PXv$, and parameters $\rho \ge 1$ and $s\ge0$, there exists 
    a codebook $\Cc_{n}$ with $M$ codewords of length $n$ whose maximal error probability satisfies
        \begin{equation}
        p_{e}(\Cc_{n}) \le \rcux_{\rho}(n,M) \le \rcux_{\rho,s}(n,M) \label{eq:EXP_Finite}
        \end{equation}
    where
        \begin{align}
        \rcux_{\rho}(n,M) & \defeq \bigg(4(M-1)\EE\bigg[\PP\Big[q^{n}(\Xvbar,\Yv)\ge q^{n}(\Xv,\Yv)\,\Big|\,\Xv,\Xvbar\Big]^{1/\rho}\bigg]\bigg)^{\rho} \label{eq:EXP_RCX}\\
        \rcux_{\rho,s}(n,M) & \defeq \Bigg(4(M-1)\EE\Bigg[\EE\Bigg[\bigg(\frac{q^{n}(\Xvbar,\Yv)}{q^{n}(\Xv,\Yv)}\bigg)^{s}\,\Bigg|\,\Xv,\Xvbar\Bigg]^{1/\rho}\Bigg]\Bigg)^{\rho}.\label{eq:EXP_RCX_s}
        \end{align}
\end{thm}
\begin{IEEEproof}
    We obtain \eqref{eq:EXP_RCX} from \eqref{eq:EXP_GalBound} by weakening the expectation as follows:
    \begin{align}
    \EE\big[p_{e,m}(\Csf)^{1/\rho}\big] 
        &\le \EE\bigg[\bigg(\sum_{\mbar\ne m}\PP\Big[q^{n}(\Xv^{(\mbar)},\Yv)\ge q^{n}(\Xv^{(m)},\Yv)\,\Big|\,\Xv^{(m)},\Xv^{(\mbar)}\Big]\bigg)^{1/\rho}\bigg] \label{eq:EXP_FiniteStep1} \\
        &\le \EE\bigg[2(M-1)\PP\Big[q^{n}(\Xvbar,\Yv)\ge q^{n}(\Xv,\Yv)\,\Big|\,\Xv,\Xvbar\Big]^{1/\rho}\bigg], \label{eq:EXP_FiniteStep2}
    \end{align}
    where \eqref{eq:EXP_FiniteStep1} follows from the union bound, and \eqref{eq:EXP_FiniteStep2} 
    follows using $M'=2M-1$ along with the inequality
    \begin{equation}
        \Big(\sum_{i}a_{i}\Big)^{1/\rho} \le \sum_{i}a_{i}^{1/\rho}, \label{eq:EXP_GalEq}
    \end{equation}
    which holds for any $\rho\ge1$. We obtain \eqref{eq:EXP_RCX_s} by applying 
    Markov's inequality to the inner probability in \eqref{eq:EXP_RCX}.
\end{IEEEproof}

Following the terminology of Polyanskiy \emph{et al.} \cite{Finite}, we refer to the bounds in 
\eqref{eq:EXP_RCX}--\eqref{eq:EXP_RCX_s} as \emph{expurgated random-coding union} (RCUX) bounds.
These bounds are computable for sufficiently symmetric setups, and are thus of independent interest
for characterizing the finite-length performance \cite{Finite}.  It should be noted that
both $\rcux_{\rho}$ and $\rcux_{\rho,s}$ extend immediately to channels with memory and
general decoding metrics (not necessarily single-letter).

The bound $\rcux_{\rho,s}$ was presented by Gallager \cite{Gallager} under ML decoding  
with $s=\frac{1}{2}$. For the random-coding ensembles we consider, it will be seen that this 
choice of $s$ is optimal for ML decoding, at least in terms of the error exponent.
However, for mismatched decoding it is important to allow for an arbitrary choice of $s\ge0$.  

The following theorem gives an asymptotic bound which follows by using Lemma 
\ref{lem:EXP_GalLemma} with a choice of $f(\cdot)$ which differs from that of Gallager.

\begin{thm} \label{thm:EXP_LogBound}
    Consider a sequence of codebooks $\Csf_{n}$ containing $M_{n}'=e^{nR}$ 
    codewords which are generated independently according to $P_{\Xv}$. Suppose that
    there exists a non-negative sequence $E(n)$ growing subexponentially in $n$ (i.e. $E(n)\doteq1$) such that
    \begin{equation}
    \PP\big[q^n(\xvbar,\Yv) \ge q^n(\xv,\Yv) \,\big|\, \Xv=\xv\big] \ge e^{-E(n)} \label{eq:EXP_TechCond}
    \end{equation}
    for all $\xv$ and $\xvbar$ on the support of $P_{\Xv}$.  Then there
    exists a sequence of codebooks $\Cc_{n}$ with $M_{n}$ codewords such that
    \begin{equation}
    \lim_{n\to\infty} \frac{1}{n}\log M_{n} = R \label{eq:EXP_PropR}
    \end{equation}
    and
    \begin{align}
    p_{e}(\Cc_{n}) &\,\,\dot{\le}\, \exp\Big(\EE[\log p_{e,m}(\Csf_{n})]\Big) \label{eq:EXP_PropPe} \\
                   &\le \exp\Big(\rho\,\EE\Big[\log \EE\big[p_{e,m}(\Csf_{n})^{1/\rho}\,\big|\,\Xv^{(m)}\big]\Big]\Big), \label{eq:EXP_JensenPe}
    \end{align} 
    where \eqref{eq:EXP_JensenPe} holds for any $\rho>0$. 
\end{thm}
\begin{proof}
    The error probability associated with the transmitted codeword $\xv$ is lower bounded
    by the left-hand side of \eqref{eq:EXP_TechCond}, where $\xvbar$ is any incorrect codeword.
    The assumption in \eqref{eq:EXP_TechCond} thus 
    implies that the function $f(p_{e,m}(\Csf))=E(n)+\log p_{e,m}(\Csf)$
    is non-negative for $m=1,\dotsc,M$.  Applying Lemma \ref{lem:EXP_GalLemma}, we obtain
    that for each $n$ and any $\eta_n > 0$ there exists a codebook $\Cc_{n}$ of size 
    $M_{n}=e^{nR}\frac{\eta_{n}}{1+\eta_{n}}$ such that
    \begin{equation}
    E(n) + \log p_{e}(\Cc_{n}) \le (1+\eta_{n})\big(E(n) + \EE[\log p_{e,m}(\Csf_{n})]\big).
    \end{equation}
    Since $\log\alpha\le0$ for $\alpha\in(0,1]$, it follows that
    \begin{equation}
    \log p_{e}(\Cc_{n}) \le \eta_{n}E(n) + \EE[\log p_{e,m}(\Csf_{n})].
    \end{equation}
    Choosing $\eta_{n}=\frac{1}{E(n)}$, we obtain \eqref{eq:EXP_PropPe},
    and the assumption that $E(n)\doteq1$ implies \eqref{eq:EXP_PropR}. 
    We obtain \eqref{eq:EXP_JensenPe} by writing $\log\alpha=\rho\log(\alpha^{1/\rho})$, 
    writing $\EE[\,\cdot\,]=\EE[\EE[\,\cdot\, | \Xv^{(m)}]]$, and applying Jensen's inequality.
\end{proof}

The assumption of Theorem \ref{thm:EXP_LogBound} is mild, allowing
ensembles for which the error probability associated with any two permissible codewords
decays nearly \emph{double}-exponentially fast.  However, it is a multi-letter 
condition, and may therefore be difficult to verify directly.  A single-letter 
sufficient condition depending only on the channel, metric and cost constraint \eqref{eq:EXP_SystemCost}
is that 
    \begin{equation}
        \lim_{\gamma\to\infty} \frac{1}{\gamma}\log\log\frac{1}{\pi(\gamma)} = 0, \label{eq:EXP_SingleLetterCond}
    \end{equation}  
where
    \begin{align}
    \pi(\gamma)  &\defeq \min_{(x,\xbar)\,:\,c(x)\le\gamma,c(\xbar)\le\gamma} \PP[Y_{x}\in\Ec(x,\xbar)] \label{eq:EXP_DefPi}\\
    \Ec(x,\xbar) &\defeq \big\{y \,:\, q(\xbar,y) \ge q(x,y)\big\},
    \end{align}
where in \eqref{eq:EXP_DefPi} we define $Y_{x} \sim W(\cdot|x)$.
Under this assumption, the probability in \eqref{eq:EXP_TechCond} is 
lower bounded by the probability that $Y_{i} \in \Ec(X_i,\Xbar_i)$ for $i=1,\dotsc,n$,
which in turn is lower bounded by $\pi(n\Gamma)^n$.  Since $n$ times a subexponential
sequence is also subexponential, the condition of Theorem \ref{thm:EXP_LogBound} follows from \eqref{eq:EXP_SingleLetterCond}.
Further discussion is given in Appendix \ref{sec:EXP_TECH_COND}, along with some examples.

From \eqref{eq:EXP_PropPe}, we can see the advantage of the expurgated
ensemble over the non-expurgated one.  The former yields the exponent corresponding to
$-\frac{1}{n}\EE[\log p_{e,m}(\Csf_{n})]$, which is higher in general than that of 
$-\frac{1}{n}\log \EE[p_{e,m}(\Csf_{n})]$ due to Jensen's inequality. 

Using L'H\^{o}pital's rule, it is easily shown that 
$\lim_{\rho\to\infty} \rho\log\EE[Z^{1/\rho}]=\EE[\log Z]$ for any random
variable $Z$.  It follows that the inequality in \eqref{eq:EXP_JensenPe} is
actually an equality in the limit as $\rho\to\infty$.  At first glance, it 
may appear that a similar argument can be used to show that \eqref{eq:EXP_GalBound}
yields the same exponent as \eqref{eq:EXP_PropPe}.  However, there is an issue
with the order of the limits of $n$ and $\rho$.  If we take $\rho\to\infty$ in 
\eqref{eq:EXP_GalBound}, the factor $2^{\rho}$ makes the right-hand
side equal $\infty$.  Letting $\rho$ grow slowly with $n$ is also 
potentially problematic, since the random variable $p_{e,m}(\Csf)$ varies with $n$.

The bounds in Theorem \ref{thm:EXP_LogBound} will prove useful for deriving 
improved exponents compared to Theorem \ref{thm:EXP_Finite}
for some codeword distributions, and for extending the type class enumeration
method beyond the finite-alphabet setting. 

\section{Expurgated Ensembles and Exponents} \label{sec:EXP_OVERVIEW}

In this section, we present an overview of various expurgated exponents and the connections between them.   
Our focus here is primarily on existing exponents or simple variations thereof, though we
also provide a dual form of the exponent in \cite{Csiszar1} which is new to the best of 
our knowledge.  Further exponents which appear for the first time in this paper
are given in Theorems \ref{thm:EXP_Dual} and \ref{thm:EXP_General} in Section \ref{sec:EXP_DISCRETE}.

Throughout the paper, we consider three expurgated ensembles, each of
which depends on an input distribution $Q$:
\begin{enumerate}
  \item The i.i.d. ensemble is characterized by
        \begin{equation}
        P_{\Xv}(\xv) = \prod_{i=1}^{n}Q(x_{i}). \label{eq:EXP_Px_IID}
        \end{equation}
        This codeword distribution is valid for both discrete and continuous alphabets, but it is
        not suitable for channels with cost constraints, since in all non-trivial cases there is
        a non-zero probability of violating the constraint.
  \item The constant-composition ensemble is characterized by
        \begin{equation}
        P_{\Xv}(\xv) = \frac{1}{|T^{n}(Q_{n})|} \openone\big\{\xv\in T^{n}(Q_{n})\big\}, \label{eq:EXP_Px_CC}
        \end{equation}
        where $Q_{n}$ is a type with the same support as $Q$ such that $|Q_n(x)-Q(x)|=O\big(\frac{1}{n}\big)$  
        for all $x$.  This codeword distribution relies on the input being finite.  It is directly applicable 
        to channels with cost constraints, since each codeword satisfies \eqref{eq:EXP_SystemCost}
        provided that $\EE_{Q_n}[c(X)]\le\Gamma$, which in turn can be achieved
        provided that $\EE_{Q}[c(X)]\le\Gamma$.
  \item The cost-constrained ensemble is characterized by
            \begin{equation}
            \PXv(\xv)=\frac{1}{\mu_{n}}\prod_{i=1}^{n}Q(x_{i})\openone\big\{\xv\in\Dc_{n}\big\}, \label{eq:EXP_Px_Multi}
            \end{equation}
        where 
            \begin{equation}
            \Dc_{n} \defeq \bigg\{\xv\,:\,\frac{1}{n}\sum_{i=1}^{n}c(x_{i}) \le \Gamma, \bigg|\frac{1}{n}\sum_{i=1}^{n}a_{l}(x_{i})-\phi_{l}\bigg|\le\frac{\delta}{n},\, l=1,\dotsc,L\bigg\}, \label{eq:EXP_SetDn} \\
            \end{equation}
        and where $\delta$ is a positive constant (independent of $n$), 
        $\{a_{l}(\cdot)\}_{l=1}^{L}$ are functions with means $\phi_{l}\defeq\EE_{Q}[a_{l}(X)]$,
        and $\mu_n$ is a normalizing constant.
        This codeword distribution is valid for both discrete and continuous alphabets, 
        and ensures that each codeword satisfies \eqref{eq:EXP_SystemCost}.
        Both $c(\cdot)$ and $\{a_{l}(\cdot)\}$ can be thought of as cost functions, and we
        will distinguish between the two by referring to them as the \emph{system cost}
        and \emph{auxiliary costs} respectively.
        In contrast to the system cost, the auxiliary costs are functions 
        which can be optimized.  That is, while the system cost is given as part of the
        problem statement, the auxiliary costs are introduced to improve the performance
        of the random-coding ensemble itself \cite{MMRevisited,PaperITA,JournalSU}.
\end{enumerate}

We proceed by stating and comparing the exponents obtained by the above ensembles;
derivations will be given in Section \ref{sec:EXP_DISCRETE}.  Except where 
stated otherwise, we assume that the channel is a DMC with unconstrained inputs.

A straightforward generalization of Gallager's i.i.d. exponent to the setting 
of mismatched decoding is as follows:
    \begin{equation}
    \Eexiid(Q,R) \defeq \sup_{\rho\ge1}\Exiid(Q,\rho)-\rho R, \label{eq:EXP_Eex_IID}
    \end{equation}
where
    \begin{equation}
    \Exiid(Q,\rho) \defeq \sup_{s\ge0}-\rho\log\sum_{x,\xbar}Q(x)Q(\xbar)\Bigg(\sum_{y}W(y|x)\bigg(\frac{q(\xbar,y)}{q(x,y)}\bigg)^{s}\Bigg)^{1/\rho}. \label{eq:EXP_Ex_IID}
    \end{equation}
The objective in \eqref{eq:EXP_Ex_IID} is concave in $s$, and under ML decoding 
(i.e. $q(x,y)=W(y|x)$), it is also unchanged when $s$ is replaced by $1-s$.  
From these properties, it follows that $s=\frac{1}{2}$ is optimal for ML decoding, and
thus the exponent is the same as that of Gallager \cite{Gallager}.

Csisz\'{a}r and K\"{o}rner \cite{Csiszar1} make use of the constant-composition
codeword distribution in \eqref{eq:EXP_Px_CC}.  The analysis is significantly different 
to that of Gallager, and yields an exponent in a different form, namely
    \begin{equation}
    \Eexcc(Q,R) \defeq \min_{\substack{P_{X\Xbar Y}\in\SetTcc(Q)\\I_P(X;\Xbar) \le R}}
    D(P_{X\Xbar Y}\|Q \times Q \times W)-R,\label{eq:EXP_PrimalAlt}
    \end{equation}
where the notation $Q \times Q \times W$ denotes the distribution $Q(x)Q(\xbar)W(y|x)$, and
    \begin{equation}
    \SetTcc(Q) \defeq \Big\{ P_{X \Xbar Y}\in\Pc(\Xc\times\Xc\times\Yc)\,:\, P_{X}=Q,P_{\Xbar}=Q,\EE_{P}[\log q(\Xbar,Y)]\ge\EE_{P}[\log q(X,Y)]\Big\}. \label{eq:EXP_SetT}
    \end{equation} 
The objective in \eqref{eq:EXP_PrimalAlt} follows from \cite[Eq. (32)]{Csiszar1}
and the identity
    \begin{equation}
    D(P_{X\Xbar Y}\|Q \times Q \times W) = D(P_{X\Xbar Y}\|P_{X\Xbar} \times W) + I_{P}(X;\Xbar), \label{eq:EXP_ChainRule}
    \end{equation}
which holds for any $P_{X\Xbar Y}$ such that $P_{X}=P_{\Xbar}=Q$.
Defining $P_{Y}(y) \defeq \sum_{x}Q(x)W(y|x)$, we observe that $\Eexcc$ 
is positive for sufficiently small $R$ provided that $\EE_{Q\times W}[\log q(X,Y)] > \EE_{Q\times P_Y}[\log q(X,Y)]$.
It was shown in \cite{Csiszar2} that the mismatched capacity is in fact zero
unless this condition holds for some $Q$.

The following theorem provides the means for comparing the above two exponents, as well
as that of \cite{ExpurgCKM}.

\begin{thm} \label{thm:EXP_Duality}
    For any input distribution $Q$ and rate $R$, we have
        \begin{align}
        \Eexcc(Q,R) & = \sup_{s\ge0} \min_{\substack{P_{X\Xbar} \,:\, P_{X}=Q, P_{\Xbar}=Q,\\ I_{P}(X;\Xbar) \le R}}
        \EE_{P}[d_{s}(X,\Xbar)] + I_{P}(X;\Xbar)-R \label{eq:EXP_Duality1}\\
        & = \sup_{\rho\ge1}\Excc(Q,\rho)-\rho R, \label{eq:EXP_Eex_CC}
        \end{align}
        where
        \begin{align}
        d_{s}(x,\xbar) &\defeq -\log\sum_{y}W(y|x)\bigg(\frac{q(\xbar,y)}{q(x,y)}\bigg)^s \label{eq:EXP_ChernoffDist} \\
        \Excc(Q,\rho)  &\defeq \sup_{s\ge0,a(\cdot)}-\rho\sum_{x}Q(x)\log\sum_{\xbar}Q(\xbar)\Bigg(\sum_{y}W(y|x)\bigg(\frac{q(\xbar,y)}{q(x,y)}\bigg)^{s}\frac{e^{a(\xbar)}}{e^{a(x)}}\Bigg)^{1/\rho}. \label{eq:EXP_Ex_CC}
        \end{align}
\end{thm}
\begin{IEEEproof}
    See Appendix \ref{sec:EXP_DUALITY_PROOF}.
\end{IEEEproof}

Equations \eqref{eq:EXP_Eex_CC} and \eqref{eq:EXP_Ex_CC} strongly resemble 
\eqref{eq:EXP_Eex_IID}--\eqref{eq:EXP_Ex_IID}.  Equation \eqref{eq:EXP_Duality1}
is a generalization of the exponent in \cite{ExpurgCKM}, which is 
recovered by setting $q(x,y)=W(y|x)$ and $s=\frac{1}{2}$.  Using the same
argument as the one following \eqref{eq:EXP_Ex_IID}, it can be shown that the latter
choice is optimal.  From the proof of Theorem \ref{thm:EXP_Duality},
this implies the optimality of $s=\frac{1}{2}$
in \eqref{eq:EXP_Ex_CC} under ML decoding, though the optimal choice of $a(\cdot)$
is unclear in general.  To our knowledge, the expression in \eqref{eq:EXP_Ex_CC}
has not appeared previously even for ML decoding.

As noted in \cite{ExpurgOmura,Csiszar1}, we can write \eqref{eq:EXP_Duality1}
in the language of rate-distortion theory \cite[Ch. 10]{Cover}.  Fix $s\ge0$ and define   
    \begin{equation}
    D_s(Q,R) \defeq \min_{\substack{P_{X\Xbar} \,:\, P_{X}=Q, P_{\Xbar}=Q,\\ I_{P}(X;\Xbar) \le R}} \EE_{P}[d_{s}(X,\Xbar)]. \label{eq:EXP_Ds}
    \end{equation}
This can be interpreted as the distortion-rate function of a source $X$ with
a reproduction variable $\Xbar$, subject to the additional constraint that
each reproduction codeword $\xvbar$ has empirical distribution $Q$.  
For any $s\ge0$, the constraint on the mutual information in \eqref{eq:EXP_Duality1} is
active for sufficiently small $R$.  The supremum of all such rates is given by
    \begin{equation}
    R_{s}(Q) \defeq I_{P^{*}}(X;\Xbar) \label{eq:EXP_Rs},
    \end{equation}
where
    \begin{equation}
    P_{X\Xbar}^{*} \defeq \argmin_{P_{X\Xbar} \,:\, P_{X}=Q, P_{\Xbar}=Q} \EE_{P}[d_{s}(X,\Xbar)] + I_{P}(X;\Xbar).
    \end{equation}
For $R\le R_{s}$ we have $I_{P}(X;\Xbar)=R$ under the minimizing $P_{X\Xbar Y}$, whereas
for $R\ge R_{s}$ the minimum in \eqref{eq:EXP_Duality1} decreases linearly with $R$
for any fixed $s$.  It follows that
\begin{equation}
    \Eexcc(Q,R) = \sup_{s\ge0} \Eexcc(Q,R,s),
\end{equation}
where
\begin{equation}
    \Eexcc(Q,R,s) \defeq \begin{cases} D_{s}(Q,R) & R \le R_{s}(Q) \\ D_{s}(Q,R_{s}) + R_{s}(Q) - R & R > R_{s}(Q). \end{cases} \label{eq:EXP_Eex_RD}
\end{equation}

By applying Jensen's inequality to \eqref{eq:EXP_Ex_CC} and setting $a(x)=0$,
we immediately obtain 
\begin{equation}
    \Eexcc(Q,R) \ge \Eexiid(Q,R). \label{eq:EXP_Comparison}
\end{equation}
It was shown in \cite[Ex. 10.18]{CsiszarBook} that \eqref{eq:EXP_Comparison} holds
with equality under ML decoding with an optimized input distribution $Q$.  However,
when either the decoding rule or input distribution is fixed,
the inequality in \eqref{eq:EXP_Comparison} can be strict; an example is given
at the end of this section.  In Section \ref{sec:EXP_DISC_SIMPLE}, we show that 
the stronger exponent $\Eexcc$, in the form given in \eqref{eq:EXP_Eex_CC}, remains achievable
in the case of continuous alphabets, with the summations in \eqref{eq:EXP_Ex_CC}
replaced by integrals.  This is proved using the cost-constrained ensemble in \eqref{eq:EXP_Px_Multi}.

The following proposition generalizes Gallager's expression 
for the expurgated exponent as $R\to0^{+}$ for channels whose zero-error capacity 
\cite{ShannonZero} is zero, and shows that the inequality in \eqref{eq:EXP_Comparison} 
becomes an equality in the limit.

\begin{prop} \label{prop:EXP_RateZero}
Fix any input distribution $Q$ such that all pairs $(x,\xbar)$ with 
$Q(x)Q(\xbar)>0$ share a common output, i.e. $W(y|x)W(y|\xbar)>0$ for some $y$.  Then
    \begin{equation}
        \lim_{R\rightarrow0^{+}}\Eexcc(Q,R)
        = \lim_{R\rightarrow0^{+}} \Eexiid(Q,R)
        = \sup_{s\geq0} \EE[d_{s}(X,\Xbar)],
    \end{equation} 
where $d_{s}$ is defined in \eqref{eq:EXP_ChernoffDist}, and the expectation is taken
with respect to $Q(x)Q(\xbar)$.
\end{prop}
\begin{IEEEproof}
    See Appendix \ref{sec:EXP_RZERO_PROOF}.
\end{IEEEproof}

We conclude this section with a numerical example. 
The channel is defined by the entries of the $|\Xc|\times|\Yc|$ matrix 
\begin{equation}
    \left[\begin{array}{ccc}
        1-2\delta_{0} & \delta_{0} & \delta_{0}\\
        \delta_{1} & 1-2\delta_{1} & \delta_{1}\\
        \delta_{2} & \delta_{2} & 1-2\delta_{2}
    \end{array}\right], \label{eq:SU_MatrixW}
\end{equation}
and the decoding metric is defined similarly with a fixed $\delta\in(0,\frac{1}{3})$
in place of each $\delta_i$ ($i=1,2,3$), yielding a minimum Hamming distance rule.
Figure \ref{fig:EXP_Exponents} plots
the exponents in the case that $\delta_{0}=0.01$, $\delta_{1}=0.05$, $\delta_{2}=0.25$
and $Q=\big(\frac{1}{3},\frac{1}{3},\frac{1}{3}\big)$.  We observe that
$\Eexcc > \Eexiid$ at all positive rates, and the gap is particularly significant
in the mismatched case.  However, consistent with Proposition
\ref{prop:EXP_RateZero}, the two coincide in the limit as $R\to0$.

As noted in \cite{Csiszar1}, if $Q$ is optimized, then the two exponents 
coincide for ML decoding. However, the strict inequality 
$\Eexcc > \Eexiid$ remains possible for other decoding rules. 

\begin{figure}
    \begin{centering}
        \includegraphics[width=0.35\paperwidth]{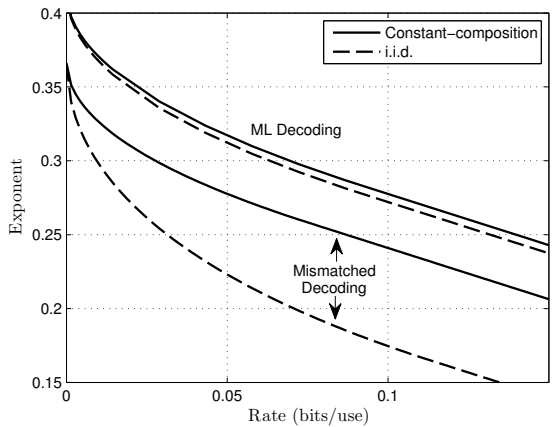}
        \par
    \end{centering}
    
    \vspace*{-3mm}
    \caption{Expurgated exponents for the channel described in \eqref{eq:SU_MatrixW}
             with minimum Hamming distance decoding and ML decoding.  The parameters are 
             $\delta_{0}=0.01$, $\delta_{1}=0.05$, $\delta_{2}=0.25$
             and $Q=\big(\frac{1}{3},\frac{1}{3},\frac{1}{3}\big)$} \label{fig:EXP_Exponents}
    \vspace*{-3mm}
\end{figure} 


\section{Derivations of the Expurgated Exponents} \label{sec:EXP_DISCRETE}

In this section, we provide several techniques for deriving the expurgated
exponents, including those introduced in Section \ref{sec:EXP_OVERVIEW} and
a further two in Theorems \ref{thm:EXP_Dual} and \ref{thm:EXP_General} below.
The approaches given here have various advantages which were outlined in Section
\ref{sec:EXP_CONTRIBUTIONS}, and which are discussed further in Section \ref{sec:EXP_DISC_CMP}.
Throughout the section, expectations are written using summations for notational simplicity 
(e.g. $\EE_Q[f(X)]=\sum_{x}Q(x)f(x)$).  However, we will highlight that certain results
apply in the case of continuous alphabets upon replacing the summations by integrals.

\subsection{Derivations Using Theorem \ref{thm:EXP_Finite}}  \label{sec:EXP_DISC_SIMPLE}

\subsubsection{i.i.d. ensemble}

We immediately obtain the exponent in \eqref{eq:EXP_Eex_IID}, as well as its generalization 
to continuous alphabets, by substituting the i.i.d. distribution in \eqref{eq:EXP_Px_IID} 
into $\rcux_{\rho,s}$ in \eqref{eq:EXP_RCX_s}.

\subsubsection{Constant-composition Ensemble}

In the case of finite alphabets, the method of types \cite[Ch. 2]{CsiszarBook} can 
be used to obtain the exact exponents corresponding to $\rcux_{\rho}$ and 
$\rcux_{\rho,s}$ for each of the ensembles defined in \eqref{eq:EXP_Px_IID}--\eqref{eq:EXP_Px_Multi}.  
The analysis is similar for each of these, so we focus on the constant-composition
ensemble described by \eqref{eq:EXP_Px_CC}.  We define
    \begin{align}
    \SetScc(Q)             &\defeq \Big\{ \Ptilde_{X\Xbar}\in\Pc(\Xc\times\Xc)\,:\, \Ptilde_{X}=Q,\Ptilde_{\Xbar}=Q\Big\} \label{eq:EXP_SetS} \\
    \SetTcc(\Ptilde_{X \Xbar}) &\defeq \Big\{ P_{X \Xbar Y}\in\Pc(\Xc\times\Xc\times\Yc)\,:\, P_{X\Xbar}=\Ptilde_{X\Xbar},\EE_{P}[\log q(\Xbar,Y)]\ge\EE_{P}[\log q(X,Y)]\Big\} \label{eq:EXP_SetT2} \\
    \SetSncc(Q)             &\defeq \SetScc(Q) \,\cap\, \Pc_{n}(\Xc\times\Xc) \\
    \SetTncc(\Ptilde_{X \Xbar}) &\defeq \SetTcc(\Ptilde_{X \Xbar})\,\cap\, \Pc_{n}(\Xc\times\Xc\times\Yc).
    \end{align}
where we overload the symbol $\SetTcc$ (see \eqref{eq:EXP_SetT}).  It follows  
that $P_{X\Xbar Y}\in \SetTcc(Q)$ (defined in \eqref{eq:EXP_SetT}) if and only if 
$P_{X\Xbar Y}\in \SetTcc(\Ptilde_{X \Xbar})$ (defined in \eqref{eq:EXP_SetT2}) 
for some $\Ptilde_{X \Xbar} \in \SetScc(Q)$. We note the following properties of types \cite[Ch. 2]{CsiszarBook}:
\begin{enumerate}
    \item For any $\Ptilde_{X\Xbar}\in\SetSncc(Q_{n})$, 
          \begin{equation}
          \PP\big[(\Xv,\Xvbar) \in T^{n}(\Ptilde_{X\Xbar})\big] \doteq e^{-nI_{\Ptilde}(X;\Xbar)}. \label{eq:EXP_Property1}
          \end{equation}
    \item If $(\xv,\xvbar) \in T^{n}(\Ptilde_{X\Xbar})$, then for any $P_{X \Xbar Y}\in\SetTncc(\Ptilde_{X\Xbar})$,
          \begin{equation}
          \PP\big[(\xv,\xvbar,\Yv) \in T^{n}(P_{X \Xbar Y})\,\big|\,\Xv=\xv\big] \doteq e^{-nD(P_{X \Xbar Y}\|\Ptilde_{X\Xbar}\times W)}. \label{eq:EXP_Property2}
          \end{equation}
\end{enumerate}

\begin{thm} \label{thm:EXP_CC1}
    Consider a discrete memoryless channel, and let the codeword distribution $P_{\Xv}$ be 
    the constant-composition distribution in \eqref{eq:EXP_Px_CC} for some input distribution $Q$. The bound
    $\rcux_{\rho}$ in \eqref{eq:EXP_RCX} satisfies the following for any rate $R>0$:
    \begin{equation}
        \inf_{\rho\ge1}\rcux_{\rho}(n,e^{nR}) \doteq e^{-n\Eexcc(Q,R)}.
    \end{equation}
\end{thm}
\begin{proof}   
    Using the codeword distribution in \eqref{eq:EXP_Px_CC} and expanding \eqref{eq:EXP_RCX} in terms of types, we obtain
        \begin{align}
        & \rcux_{\rho}(n,M)^{1/\rho} \nonumber \\
        &= 4(M-1)\sum_{\Ptilde_{X \Xbar}\in\SetSncc(Q_{n})}\PP\Big[(\Xv,\Xvbar)\in T^{n}(P_{X\Xbar})\Big]\sum_{P_{X \Xbar Y}\in\SetTncc(\Ptilde_{X\Xbar})}\PP\Big[(\xv,\xvbar,\Yv)\in T^{n}(P_{X \Xbar Y}) \,\Big|\, \Xv=\xv\Big]^{1/\rho} \label{eq:EXP_PrimalCC1} \\
        & \doteq M\max_{\Ptilde_{X \Xbar}\in\SetSncc(Q_{n})}\max_{P_{X \Xbar Y}\in\SetTncc(\Ptilde_{X\Xbar})}\exp\Big(-nI_{\Ptilde}(X;\Xbar)\Big)\exp\Big(-n\cdot\frac{1}{\rho}D\big(P_{X \Xbar Y}\|\Ptilde_{X\Xbar}\times W\big)\Big) \label{eq:EXP_PrimalCC2} \\
        & \doteq M\max_{P_{X \Xbar Y}\in\SetTcc(Q)}\exp\Big(-n\Big(\frac{1}{\rho}D\big(P_{X \Xbar Y}\|P_{X\Xbar}\times W\big)+I_{P}(X;\Xbar)\Big)\Big), \label{eq:EXP_PrimalCC3}
        \end{align}
    where in \eqref{eq:EXP_PrimalCC1} we define $(\xv,\xvbar)$ to be an arbitrary pair 
    with joint type $\Ptilde_{X\Xbar}$, \eqref{eq:EXP_PrimalCC2} follows from the properties 
    of types in \eqref{eq:EXP_Property1}--\eqref{eq:EXP_Property2} and the fact that the number of
    joint types is polynomial in $n$, and \eqref{eq:EXP_PrimalCC3} follows from the definitions 
    of $\SetSncc$, $\SetTncc$ and $\SetTcc$, and by using a standard continuity argument to
    expand the maximization from types to general distributions (e.g. see \cite{DyachkovCC}).  
    We thus obtain the exponent
        \begin{align}
        \sup_{\rho \ge 1} &\min_{P_{X \Xbar Y}\in\SetTcc(Q)}  D(P_{X \Xbar Y}\|P_{X\Xbar}\times W) + \rho\big(I_{P}(X;\Xbar)-R\big) \label{eq:EXP_Etilde1} \\ 
                          &=\min_{P_{X \Xbar Y}\in\SetTcc(Q)} \sup_{\rho \ge 1} D(P_{X \Xbar Y}\|P_{X\Xbar}\times W) + \rho\big(I_{P}(X;\Xbar)-R\big), \label{eq:EXP_Etilde2}
        \end{align}
    where \eqref{eq:EXP_Etilde2} follows from Fan's minimax theorem \cite{Minimax}, the 
    conditions of which are satisfied here since the objective is linear in $\rho$ and convex
    in $P_{X \Xbar Y}$.  Using
        \begin{equation}
        \sup_{\rho \ge 1} \rho\alpha = \begin{cases} \infty & \alpha > 0 \\ \alpha & \alpha \le 0 \end{cases} \label{eq:EXP_Etilde3}
        \end{equation}
    and the identity in \eqref{eq:EXP_ChainRule}, it follows that \eqref{eq:EXP_Etilde2} coincides with \eqref{eq:EXP_PrimalAlt}.
\end{proof}

The preceding derivation of $\Eexcc$ provides a simple alternative to that of
Csisz\'{a}r and K\"{o}rner \cite{Csiszar1}, while yielding the exponent in the same form.

\subsubsection{Cost-constrained Ensemble}

Here we provide a derivation of $\Eexcc$ in the form given in \eqref{eq:EXP_Ex_CC},
as well as its generalization to continuous alphabets, using the cost-constrained 
ensemble in \eqref{eq:EXP_Px_Multi}.  We allow for a system cost constraint 
of the form given in \eqref{eq:EXP_SystemCost}.  A key property of the ensemble 
which will prove useful in the derivations is
    \begin{equation}
        \xv \in \Dc_n \implies e^{r\big(\sum_{i=1}^{n}a(x_{i})-n\phi_{a}\big)}e^{|r|\delta} \ge 1, \label{eq:EXP_CostProperty2} 
    \end{equation}
which holds for any real number $r$, and follows immediately from 
\eqref{eq:EXP_SetDn}.  Furthermore, we have the following.

\begin{prop} \emph{\cite[Prop. 1]{JournalSU}} \label{prop:EXP_SubExp}
    Fix any input distribution $Q$ and set of cost functions $\{a_{l}\}_{l=1}^{L}$
    such that $E_{Q}[c(X)]\le\Gamma$, $E_{Q}[c(X)^2] < \infty$ and 
    $E_{Q}[a_{l}(X)^2] < \infty$ for $l=1,\dotsc,L$.  Then the normalizing
    constant $\mu_{n}$ in \eqref{eq:EXP_Px_Multi} satisfies 
        \begin{equation}
        \lim_{n\to\infty}\frac{1}{n}\log\mu_{n} = 0. \label{eq:EXP_CostProperty1}
        \end{equation}
\end{prop}

The following theorem gives an achievable error exponent for a fixed set of auxiliary costs.

\begin{thm} \label{thm:EXP_Dual}
    Consider a memoryless (possibly continuous) channel, 
    and fix any input distribution $Q$ and functions $\al$ satisfying the 
    assumptions of Proposition \ref{prop:EXP_SubExp}.  Under the cost-constrained distribution 
    in \eqref{eq:EXP_Px_Multi}, we have
        \begin{equation}
        \inf_{\rho\ge1,s\ge0}\rcux_{\rho,s}(n,e^{nR}) \,\,\dot{\le}\, e^{-n\Eexcost(Q,R,\{a_{l}\})}
        \end{equation}
    for any rate $R>0$, where
        \begin{equation} 
        \Eexcost(Q,R,\{a_{l}\}) \defeq\sup_{\rho\ge1}\Excost(Q,\rho,\{a_{l}\})-\rho R, \label{eq:EXP_Eex_Cost}
        \end{equation}
    and\footnote{In the case of continuous alphabets, the summations over sequences should be replaced by integrals.}
        \begin{equation} 
        \Excost(Q,R,\{a_{l}\}) \defeq \sup_{s\ge0,\{r_{l}\},\{\rbar_{l}\}}-\rho\log\sum_{x,\xbar}Q(x)Q(\xbar)\Bigg(\sum_{y}W(y|x)\bigg(\frac{q(\xbar,y)}{q(x,y)}\bigg)^{s}\frac{e^{\sum_{l=1}^{L}\rbar_{l}(a_{l}(\xbar)-\phi_{l})}}{e^{\sum_{l=1}^{L}r_{l}(a_{l}(x)-\phi_{l})}}\Bigg)^{1/\rho}. \label{eq:EXP_Ex_Cost}
        \end{equation}
\end{thm}
\begin{IEEEproof}
    Let $a_{l}^{n}(\xv) \defeq \sum_{i=1}^{n}a_{l}(x_{i})$ and
    $Q^{n}(\xv) \defeq \prod_{i=1}^{n}Q(x_i)$.   We start with \eqref{eq:EXP_RCX_s}, and write 
        \begin{align}
        \rcux_{\rho,s}(n,M)^{1/\rho} & = 4(M-1)\sum_{\xv,\xvbar}\PXv(\xv)\PXv(\xvbar)\Bigg(\sum_{\yv}W^{n}(\yv|\xv)\bigg(\frac{q^{n}(\xvbar,\yv)}{q^{n}(\xv,\yv)}\bigg)^{s}\Bigg)^{1/\rho}\\
        & \,\,\dot{\le}\, M\sum_{\xv,\xvbar}\PXv(\xv)\PXv(\xvbar)\Bigg(\sum_{\yv}W^{n}(\yv|\xv)\bigg(\frac{q^{n}(\xvbar,\yv)}{q^{n}(\xv,\yv)}\bigg)^{s}\frac{e^{\sum_{l=1}^{L}\rbar_{l}(a_{l}^{n}(\xvbar)-n\phi_{l})}}{e^{\sum_{l=1}^{L}r_{l}(a_{l}^{n}(\xv)-n\phi_{l})}}\Bigg)^{1/\rho} \label{eq:EXP_DualCost2}\\
        & \,\,\dot{\le}\, M\sum_{\xv,\xvbar}Q^{n}(\xv)Q^{n}(\xvbar)\Bigg(\sum_{\yv}W^{n}(\yv|\xv)\bigg(\frac{q^{n}(\xvbar,\yv)}{q^{n}(\xv,\yv)}\bigg)^{s}\frac{e^{\sum_{l=1}^{L}\rbar_{l}(a_{l}^{n}(\xvbar)-n\phi_{l})}}{e^{\sum_{l=1}^{L}r_{l}(a_{l}^{n}(\xv)-n\phi_{l})}}\Bigg)^{1/\rho},\label{eq:Exp_DualCost3}
        \end{align}
    where \eqref{eq:EXP_DualCost2} holds for any $\{r_l\}$ and $\brl$ from \eqref{eq:EXP_CostProperty2},
    and \eqref{eq:Exp_DualCost3} follows  from \eqref{eq:EXP_Px_Multi} and Proposition \ref{prop:EXP_SubExp}.  
    The proof is concluded by expanding each term in \eqref{eq:Exp_DualCost3} as a product from 
    $1$ to $n$ and optimizing $\rho$, $s$, $\{r_l\}$ and $\brl$. 
\end{IEEEproof}

We now show that we can recover $\Eexcc$ from $\Eexcost$ upon setting $L=2$ and 
optimizing the auxiliary costs; an analogous statement was shown to be true
for the random-coding exponent in \cite{JournalSU}. Setting
$\rbar_{1}=r_{2}=1$ and $\rbar_{2}=r_{1}=0$, and optimizing 
$a_{1}(\cdot)$ and $a_{2}(\cdot)$, we obtain
    \begin{align}
        \Excost(Q,\rho) &= \sup_{s\ge0,a_{1}(\cdot), a_{2}(\cdot)}-\rho\log\sum_{x,\xbar}Q(x)Q(\xbar)\Bigg(\sum_{y}W(y|x)\bigg(\frac{q(\xbar,y)}{q(x,y)}\bigg)^{s}\frac{e^{a_1(\xbar)-\phi_{1}}}{e^{a_2(x)-\phi_{2}}} \Bigg)^{1/\rho} \label{eq:EXP_CCStep1} \\
        & \le \sup_{s\ge0,a_{1}(\cdot), a_{2}(\cdot)}-\rho\sum_{x}Q(x)\log\sum_{\xbar}Q(\xbar)\Bigg(\sum_{y}W(y|x)\bigg(\frac{q(\xbar,y)}{q(x,y)}\bigg)^{s}\frac{e^{a_1(\xbar)-\phi_{1}}}{e^{a_2(x)-\phi_{2}}} \Bigg)^{1/\rho}, \label{eq:EXP_CCStep2}
    \end{align}
where \eqref{eq:EXP_CCStep2} follows from Jensen's inequality.  For any $s$ and $a_1(\cdot)$, 
there exists a choice of $a_2(\cdot)$ that makes Jensen's inequality hold with equality 
in $\eqref{eq:EXP_CCStep2}$, and hence the same is true after taking the supremum.
Hence, and by writing 
    \begin{equation}
    -\sum_{x}Q(x)\log\bigg(\frac{e^{-\phi_{1}}}{e^{a_2(x)-\phi_{2}}}\bigg)^{1/\rho} = -\sum_{x}Q(x)\log\big(e^{-a_{1}(x)}\big)^{1/\rho} = \frac{\phi_1}{\rho}, \label{eq:EXP_CCStep3}
    \end{equation}
we see that the $a_{2}(\cdot)$ achieving the supremum in \eqref{eq:EXP_CCStep1} is the one yielding 
equality in \eqref{eq:EXP_CCStep2}.  Renaming $a_{1}(\cdot)$ as $a(\cdot)$ and using the first equality
in \eqref{eq:EXP_CCStep3}, we obtain \eqref{eq:EXP_Ex_CC}.

It should be noted that, in accordance with Proposition \ref{prop:EXP_SubExp}, 
the supremum over $s$ and $a(\cdot)$ in \eqref{eq:EXP_Ex_CC} is restricted to
choices such that $E_{Q}[a(X)^2]<\infty$, and such that $E_{Q}[a_{2}(X)^2]<\infty$
for the choice of $a_{2}(\cdot)$  which makes Jensen's inequality hold with 
equality in \eqref{eq:EXP_CCStep2} (expressed in terms of $s$ and $a(\cdot)$).
This may rule out some parameters in the case of infinite or continuous alphabets.

While the parameters $\{r_l\}$ and $\brl$ are not necessary for obtaining \eqref{eq:EXP_CCStep1},
they can improve the exponent for a given set of auxiliary costs \cite{JournalSU}.  
That is, the more general exponent of Theorem \ref{thm:EXP_Dual} serves as an indicator of 
the performance when the auxiliary costs are chosen suboptimally.
Using a similar argument to that of \eqref{eq:EXP_CCStep1}--\eqref{eq:EXP_CCStep3}, 
it is easily shown that $\Eexcost\le\Eexcc$, and hence one 
cannot improve on the exponent obtained using $L=2$ optimally chosen auxiliary costs.

\subsection{Derivation Using Type Class Enumerators} \label{sec:EXP_TYPE_ENUM}

In the proof of Theorem \ref{thm:EXP_CC1}, we gave an exponentially tight analysis of
$\rcux_{\rho}$.  In this subsection, we show that an exponentially tight analysis can
be provided starting from an earlier step using the method of  
type class enumeration (e.g. see \cite{MerhavErasure,MerhavIC,MerhavPhysics}). 
Once again, the analysis is similar for each of the ensembles in
\eqref{eq:EXP_Px_IID}--\eqref{eq:EXP_Px_Multi}, so we focus on the 
constant-composition ensemble described by \eqref{eq:EXP_Px_CC}.

Substituting \eqref{eq:EXP_FiniteStep1} into \eqref{eq:EXP_GalBound} and defining
\begin{equation}
    d_{q}(\xv,\xvbar) \defeq -\log\PP\big[q^{n}(\xvbar,\Yv)\ge q^{n}(\xv,\Yv) \,\big|\, \Xv=\xv\big], \label{eq:EXP_dq}
\end{equation}
we obtain the bound
\begin{align}
    p_{e}(\Cc) &\le \big(2A_{n}(R,\rho)\big)^{\rho},  \label{eq:EXP_TypeEnumStart}
\end{align}
where
\begin{equation}
    A_{n}(R,\rho) \triangleq \EE\Bigg[\bigg(\sum_{\mbar\ne m}e^{-d_{q}(\Xv^{(m)},\Xv^{(\mbar)})}\bigg)^{1/\rho}\Bigg].
\end{equation}
This bound provides the starting point for our analysis.  Note that since
we have not used the inequality in \eqref{eq:EXP_GalEq}, we may allow
for $\rho\ge0$ rather than just $\rho\ge1$.

\begin{thm} \label{thm:EXP_TypeEnum}
    Consider a discrete memoryless channel, and let the codeword distribution $P_{\Xv}$ be 
    the constant-composition distribution in \eqref{eq:EXP_Px_CC} for some input distribution $Q$.  
    Then the following holds for any rate $R>0$:
    \begin{equation}
        \inf_{\rho\ge0}\big(2A_{n}(R,\rho)\big)^{\rho} \doteq e^{-n\Eexcc(Q,R)}.
    \end{equation}
\end{thm}
\begin{proof}
    For $m=1,\dotsc,M$ and each joint type $\Ptilde_{X\Xbar}$, we define the random variable
        \begin{equation}
        N_m(\Ptilde_{X\Xbar}) \defeq \sum_{\mbar \ne m} \openone\big\{ (\Xv^{(m)},\Xv^{(\mbar)}) \in T^{n}(\Ptilde_{X\Xbar}) \big\}. \label{eq:EXP_Nm}
        \end{equation} 
    Under the random-coding distribution in \eqref{eq:EXP_Px_CC}, we have 
    $N_m(\Ptilde_{X\Xbar})=0$ with probability one if $\Ptilde_{X\Xbar}\notin\SetSncc(Q_n)$.
    That is, the marginal distribution of each codeword must agree with $Q$. Since $d_{q}$ depends 
    only on the joint type of its arguments, we define 
    $d_{q}(\Ptilde_{X\Xbar}) \defeq \frac{1}{n}d_{q}(\xv,\xvbar)$, 
    where $(\xv,\xvbar) \in T^{n}(\Ptilde_{X\Xbar})$.
    
    Making repeated use of the fact that the number of joint
    types is polynomial in $n$, we have the following:
        {\allowdisplaybreaks 
        \begin{align}
        A_n(R,\rho)
        &= \EE\Bigg[\bigg(\sum_{\Ptilde_{X\Xbar}}N_m(\Ptilde_{X\Xbar})e^{-n d_{q}(\Ptilde_{X\Xbar})}\bigg)^{1/\rho}\Bigg]\\
        &\doteq \EE \bigg[\max_{\Ptilde_{X\Xbar}}N_m(\Ptilde_{X\Xbar})^{1/\rho}e^{-n d_{q}(\Ptilde_{X\Xbar})/\rho}\bigg]\\
        &\doteq \EE\bigg[\sum_{\Ptilde_{X\Xbar}}N_m(\Ptilde_{X\Xbar})^{1/\rho}e^{-n d_{q}(\Ptilde_{X\Xbar})/\rho}\bigg]\\
        &\doteq \max_{\Ptilde_{X\Xbar}}\EE\Big[N_m(\Ptilde_{X\Xbar})^{1/\rho}\Big]e^{-n d_{q}(\Ptilde_{X\Xbar})/\rho}, \label{eq:EXP_AnBound}
        \end{align} }
    where \eqref{eq:EXP_AnBound} follows by first taking the summation outside the expectation.
    It follows from \eqref{eq:EXP_AnBound} that
        \begin{equation}
        \big(2A_n(R,\rho)\big)^{\rho} \doteq \max_{\Ptilde_{X\Xbar}}\bigg(\EE\Big[N_m(\Ptilde_{X\Xbar})^{1/\rho}\Big]\bigg)^\rho e^{-n d_{q}(\Ptilde_{X\Xbar})}. \label{eq:EXP_AnRhoBound}
        \end{equation}
    Similarly to \cite[Eq. (34)]{MerhavErasure}, we have for 
    all $\Ptilde_{X\Xbar}\in\SetSncc(Q_n)$ that
        \begin{equation}
        \EE\Big[N_m(\Ptilde_{X\Xbar})^{1/\rho}\Big] \doteq
        \left\{\begin{array}{ll}
               \exp\big(n\big(R-I_{\Ptilde}(X;\Xbar)\big)\big) & R < I_{\Ptilde}(X;\Xbar)\\
               \exp\big(n\big(R-I_{\Ptilde}(X;\Xbar)\big)/\rho\big) & R \ge I_{\Ptilde}(X;\Xbar).
               \end{array}
        \right.
        \end{equation}
    This follows from the fact that given $\Xv^{(m)}=\xv$,
    $N_m(\Ptilde_{X\Xbar})$ is the sum of $e^{nR}-1$ binary independent 
    random variables, 
        \begin{equation}
        U_{\mbar}\defeq\openone\Big\{(\xv,\Xv^{(\mbar)}) \in T^{n}(\Ptilde_{X\Xbar})\Big\}, \label{eq:EXP_Um}
        \end{equation}
    whose expectations are of the exponential order of $e^{-nI_{\Ptilde}(X;\Xbar)}$ (see \eqref{eq:EXP_Property1}).
    Furthermore, expanding \eqref{eq:EXP_dq} in terms of types and using the property in
    \eqref{eq:EXP_Property2}, we obtain
        \begin{align}
        e^{-n d_{q}(\Ptilde_{X\Xbar})} &\doteq \exp\bigg(-n \min_{P_{X \Xbar Y}\in\SetTcc(\Ptilde_{X\Xbar})} D\big(P_{X \Xbar Y}\|\Ptilde_{X\Xbar}\times W\big)\bigg) \\
                                                 &\defeq e^{-n D_{q}(\Ptilde_{X\Xbar})}. \label{eq:EXP_DefF}   
        \end{align}
    Upon taking into account all the possible empirical distributions
    $\{\Ptilde_{X\Xbar}\}$ in \eqref{eq:EXP_Um}, we obtain
        \begin{equation}
        \big(2A_n(R,\rho)\big)^{\rho} \doteq e^{-n\min\{E_1(R,\rho),E_2(R)\}}, \label{eq:EXP_AnExponent}
        \end{equation}
    where
        \begin{equation}
        E_{1}(R,\rho) \defeq \min_{\substack{\Ptilde_{X\Xbar}\in\SetScc(Q)\\I_{\Ptilde}(X;\Xbar)\ge R}} D_q(\Ptilde_{X\Xbar})+\rho\big(I_{\Ptilde}(X;\Xbar)-R\big) \label{eq:EXP_E1}
        \end{equation}
    and
        \begin{equation}
        E_{2}(R) = \min_{\substack{\Ptilde_{X\Xbar}\in\SetScc(Q)\\I_{\Ptilde}(X;\Xbar)\le R}} D_q(\Ptilde_{X\Xbar}) + I_{\Ptilde}(X;\Xbar) - R. \label{eq:EXP_E2}
        \end{equation}
    Combining \eqref{eq:EXP_ChainRule}, \eqref{eq:EXP_DefF} and \eqref{eq:EXP_E2}, 
    we see that $E_{2}(R)$ coincides with $\Eexcc$ in the form given in \eqref{eq:EXP_PrimalAlt}.
    It remains to show that $E_1(R,\rho)$, for the optimum choice of $\rho$, is
    never smaller than $E_2(R)$.  This can be seen by noting that
    since \eqref{eq:EXP_E1} contains the constraint $I_{\Ptilde}(X;\Xbar)\ge R$,
    the term multiplying $\rho$ in \eqref{eq:EXP_E1} is non-negative.
    Thus, the best choice of $\rho$ is to take the limit as $\rho\to\infty$, and
    hence the minimum in \eqref{eq:EXP_E1} is achieved by some $\Ptilde_{X\Xbar}$ 
    satisfying $I_{\Ptilde}(X;\Xbar)=R$.  Since this joint distribution also
    satisfies the constraints in \eqref{eq:EXP_E2}, we conclude that $E_{1} \ge E_{2}$,
    thus completing the proof.
\end{proof} 

While the exponents of Theorems \ref{thm:EXP_CC1} and \ref{thm:EXP_TypeEnum} coincide
for the constant-composition ensemble, the type enumeration approach can yield
strictly higher exponents for other codeword distributions; see Section 
\ref{sec:EXP_DISC_CMP} for details.  

\subsection{Derivation Using Distance Enumerators} \label{sec:EXP_DIST_ENUM}

In this subsection, we extend the preceding type enumeration analysis
to channels with infinite or continuous alphabets, and then discuss the
further extension to channels with memory. We make use of Theorem \ref{thm:EXP_LogBound},
and we assume that the technical assumption therein is satisfied
(see Appendix \ref{sec:EXP_TECH_COND} for discussion).  
We fix $s\ge0$ and make use of $d_s$ in \eqref{eq:EXP_ChernoffDist}
(or its counterpart for continuous outputs with an integral in place of the summation), as 
well as its multi-letter extension
    \begin{equation}
    d_{s}^{n}(\xv,\xvbar) \defeq \sum_{i=1}^{n} d_{s}(x_i,\xbar_i). \label{eq:EXP_ChernoffDMulti}
    \end{equation}

\begin{thm} \label{thm:EXP_General}
    Consider a memoryless (possibly continuous) channel, and fix any codeword distribution 
    $P_{\Xv}$ satisfying the assumption of Theorem \ref{thm:EXP_LogBound}.  The exponent
    \begin{equation}
        \Eex(R) \triangleq \EE\bigg[\inf_{D\,:\,R(D,\Xv) \le R} D + R(D,\Xv)-R\bigg] \label{eq:EXP_FinalExpGen}
    \end{equation} 
    is achievable for any function $R(D,\xv)$ such that 
    $\PP\big[d_{s}^{n}(\xv,\Xvbar) < nD] \,\,\dot{\le}\, e^{-nR(D,\xv)}$ uniformly in $\xv$,
    and such that $R(\cdot,\xv)$ is continuous for any given $\xv$.
\end{thm}
\begin{proof}
    We claim that \eqref{eq:EXP_JensenPe} implies the following analog of 
    \eqref{eq:EXP_TypeEnumStart} for a sequence of codebook $\Cc_n$ of rate approaching $R$:
        \begin{equation}
            p_{e}(\Cc_n) \,\,\dot{\le}\, \exp\Big(\rho\,\EE\big[\log A_{n}(R,\rho,\Xv^{(m)})\big]\Big), \label{eq:EXP_DistEnumStart2}
        \end{equation}
    where
        \begin{equation}
        A_{n}(R,\rho,\Xv^{(m)}) \defeq \EE\Bigg[\bigg(\sum_{\mbar\ne m}e^{[-d_{s}^{n}(\Xv^{(m)},\Xv^{(\mbar)})]^{+}}\bigg)^{1/\rho}\,\bigg|\,\Xv^{(m)}\Bigg]. \label{eq:EXP_DistEnumStart}
        \end{equation}
    In the absence of the $[\cdot]^{+}$ function in the exponent, this follows
    directly from the union bound and Markov's inequality, similarly to
    the proof of Theorem \ref{thm:EXP_Finite}. The introduction of
    the $[\cdot]^{+}$ function corresponds to instead taking the better of Markov's 
    inequality and the trivial bound $\PP[\cdot] \le 1$.\footnote{This analysis corrects an error in the conference version of this work \cite{PaperIZS}, where the $[\cdot]^{+}$ function was omitted.  This omission does not affect the analysis for ML decoding, since the Bhattacharyya distance is non-negative.  However, in general, the function $d_s(\cdot,\cdot)$ may be negative.}

    For a fixed transmitted codeword $\Xv^{(m)}=\xv$, we analyze $A_{n}(R,\rho,\xv)$
    using \emph{distance enumerators}:
        \begin{equation}
        \sum_{\mbar\ne m}e^{-[d_{s}^{n}(\xv,\Xv^{(\mbar)})]^{+}} \le \sum_{k=0}^{\infty}e^{-nk\delta}N_{m}(k,\xv),
        \end{equation}
    where $\delta>0$ is arbitrary, and
        \begin{align}
        N_{m}(0,\xv) &\defeq \sum_{\mbar\ne m} \openone\big\{d_{s}^{n}(\xv,\Xv^{(\mbar)}) < n\delta \big\}  \\
        N_{m}(k,\xv) &\defeq \sum_{\mbar\ne m} \openone\big\{nk\delta \le d_{s}^{n}(\xv,\Xv^{(\mbar)}) < n(k+1)\delta \big\} \quad (k\ge1). \label{eq:EXP_GEN_Nm}
        \end{align} 
        
    Using Markov's inequality, we can upper-bound the left-hand side of \eqref{eq:EXP_TechCond} 
    by $e^{-d_{s}^{n}(\xv,\xvbar)}$.  It thus follows from the assumption of Theorem \ref{thm:EXP_LogBound}
    that the highest value of $k$,
        \begin{equation}
        k_{\mathrm{max}}(n) \defeq \max_{\xv \,:\, P_{\Xv}(\xv)>0}\max\big\{k\,:\,\PP\big[N_{m}(k,\xv) > 0\big] \ne 0\big\}, \label{eq:EXP_kmax}
        \end{equation}
    grows subexponentially in $n$ for all $s\ge0$.  
    Thus, analogously to \eqref{eq:EXP_AnRhoBound}, the quantity $A_{n}(R,\rho,\xv)$ 
    defined in \eqref{eq:EXP_DistEnumStart} satisfies
        \begin{equation}
        A_{n}(R,\rho,\xv)^{\rho} \,\,\dot{\le}\, \max_{k\ge0}\Big(\EE\big[N_{m}(k,\xv)^{1/\rho}\big]\Big)^{\rho}e^{-nk\delta}. 
        \end{equation}
    We further upper bound this expression by removing the lower inequality in the indicator function in 
    \eqref{eq:EXP_GEN_Nm}. The key issue is now to assess the exponential rate of decay of the binary random variable
        \begin{equation}
            U_{\mbar}(\xv) \defeq \openone\big\{d_{s}^{n}(\xv,\Xv^{(\mbar)}) < n(k+1)\delta \big\}
        \end{equation}
    for a given transmitted codeword $\xv$, i.e. to find the exponent of 
    $\PP\big[d_{s}^{n}(\xv,\Xvbar) < nD]$.  This can be done
    using standard large deviations techniques such as the Chernoff bound.
    Letting $R(D,\xv)$ be as defined in the theorem statement, we have 
    similarly to \eqref{eq:EXP_AnExponent} that
        \begin{equation}
        A_{n}(R,\rho,\xv)^{\rho} \,\,\dot{\le}\, e^{-n\min\{E_1(R,\rho,\delta,\xv),E_2(R,\delta,\xv)\}}, \label{eq:EXP_AnExponent2}
        \end{equation}
    where
        \begin{align}
            E_{1}(R,\rho,\delta,\xv) &\defeq \min_{k\,:\,R((k+1)\delta,\xv) \ge R} k\delta + \rho\big(R((k+1)\delta,\xv)-R\big) \\
            E_{2}(R,\delta,\xv) &\defeq \min_{k\,:\,R((k+1)\delta,\xv) \le R} k\delta + R((k+1)\delta,\xv)-R.
        \end{align}
    Upon taking the limit $\delta\to0$ and using the assumption that $R(\cdot,\xv)$ is lower semicontinuous, these become
        \begin{align}
            E_{1}(R,\rho,\xv) &\defeq \inf_{D\,:\,R(D,\xv) \ge R} D + \rho\big(R(D,\xv)-R\big) \\
            E_{2}(R,\xv) &\defeq \inf_{D\,:\,R(D,\xv) \le R} D + R(D,\xv)-R. \label{eq:EXP_E2Gen}
        \end{align}
    Analogously to Section \ref{sec:EXP_TYPE_ENUM}, the optimal choice of $\rho$ is
    in the limit as $\rho\to\infty$, and we obtain $E_{2} \le E_{1}$, and hence
        \begin{equation}
        \inf_{\rho\ge0} A_{n}(R,\rho,\xv)^{\rho} \,\,\dot{\le}\, e^{-n E_2(R,\xv)}. \label{eq:EXP_AnExponent3}
        \end{equation}
    Substituting \eqref{eq:EXP_AnExponent3} into \eqref{eq:EXP_DistEnumStart2}, we obtain
    $p_{e}(\Cc) \,\,\dot{\le}\, e^{-n\EE[E_2(R,\Xv)]}$,
    thus yielding \eqref{eq:EXP_FinalExpGen}.
\end{proof}

After a suitable modification of the definition of $d_{s}^{n}(\xv,\xvbar)$, \eqref{eq:EXP_FinalExpGen}
extends immediately to more general channels and metrics (e.g. channels with memory).
The ability to simplify the exponent (e.g. to a single-letter expression) depends
on the form of $R(D,\xv)$, which in turn depends strongly on the codeword distribution $P_{\Xv}$.
In some cases, $P_{\Xv}$ can be chosen in such a way that $R(D,\xv)$ is the same for
all $\xv$ with $P_{\Xv}(\xv)>0$, thus greatly simplifying \eqref{eq:EXP_FinalExpGen}.

In Appendix \ref{sec:EXP_GEN_ANALYSIS}, we particularize Theorem \ref{thm:EXP_General}
to the cost-constrained ensemble with a single auxiliary cost $a_{1}(x)=a(x)$, 
and show that after optimizing $a(\cdot)$, \eqref{eq:EXP_FinalExpGen} yields 
the exponent $\Eexcc(Q,R)$ in \eqref{eq:EXP_Eex_CC}.  
In accordance with Proposition \ref{prop:EXP_SubExp}, 
we require the auxiliary cost $a(\cdot)$ to satisfy $\EE_Q[a(X)^2]<\infty$.

\subsection{Comparison of Techniques} \label{sec:EXP_DISC_CMP}

For the constant-composition codeword distribution, the
approaches of Sections \ref{sec:EXP_DISC_SIMPLE} and \ref{sec:EXP_TYPE_ENUM}
led to the same exponent, namely $\Eexcc$.  It should be noted, however, 
that the type enumeration approach can yield a strictly higher 
exponent than that of $\rcux_{\rho}$ in Theorem \ref{thm:EXP_Finite} for some codeword distributions.  
Here we discuss the simple example of the i.i.d. distribution in \eqref{eq:EXP_Px_IID}.  
Applying properties of types in the same way in Section \ref{sec:EXP_DISC_SIMPLE}, 
it is easily verified that the exponent of $\rcux_{\rho}$  is
    \begin{equation}
    \min_{\substack{P_{X\Xbar Y} \,:\, D(P_{X\Xbar}\|Q \times Q) \le R, \\ \EE_{P}[\log q(\Xbar,Y)]\ge\EE_{P}[\log q(X,Y)]}}
    D(P_{X\Xbar Y}\|Q\times Q\times W)-R.  \label{eq:EXP_PrimalIID}
    \end{equation}
On the other hand, the analysis of Section \ref{sec:EXP_TYPE_ENUM} yields
an exponent of the same form as \eqref{eq:EXP_PrimalIID} with an additional 
constraint $P_{X}=Q$ in the minimization.  To see this, we note that the 
quantity $N_{m}(\Ptilde_{X\Xbar})$ defined in \eqref{eq:EXP_Nm} satisfies
    \begin{align}
    \EE\Big[N_m(\Ptilde_{X\Xbar})^{1/\rho}\Big] &= \PP\big[\Xv^{(m)}\in T^{n}(\Ptilde_{X})\big]\EE\Big[N_m(\Ptilde_{X\Xbar})^{1/\rho}\,\Big|\,\Xv^{(m)}\in T^{n}(\Ptilde_{X})\Big] \\
    &\doteq \left\{\begin{array}{ll}
            \exp\big(-nD(\Ptilde_{X}\|Q)\big)\cdot\exp\big(n\big(R-D(\Ptilde_{X\Xbar}\|\Ptilde_{X} \times Q)\big)\big) & R < I_{\Ptilde}(X;\Xbar)\\
            \exp\big(-nD(\Ptilde_{X}\|Q)\big)\cdot\exp\big(n\big(R-D(\Ptilde_{X\Xbar}\|\Ptilde_{X} \times Q)\big)/\rho\big) & R \ge I_{\Ptilde}(X;\Xbar).
            \end{array}
            \right.
    \end{align}
The additional factor $\exp\big(-nD(\Ptilde_{X}\|Q)\big)$ leads to an additive
$\rho D(\Ptilde_{X}\|Q)$ term in the exponent $E_{2}$ in \eqref{eq:EXP_E2}.
The optimal choice of $\rho$ is again in the limit as $\rho\to\infty$, and under
this choice the minimizing $\Ptilde_{X\Xbar}$ must satisfy $\Ptilde_{X}=Q$ so that
the divergence is forced to zero.

Depending on the channel, metric and input distribution, adding the constraint 
$P_X=Q$ to \eqref{eq:EXP_PrimalIID} may yield a strict improvement in the exponent. 
Since both derivations are exponentially tight from the step at which they start,
we conclude that the weakness of the simpler derivation is in the inequality in 
\eqref{eq:EXP_FiniteStep2}, or more precisely, the use of \eqref{eq:EXP_GalEq}.  While this
step simplifies the derivations, the above example shows that it is
not exponentially tight in general.

Another approach to recovering the constraint $\Ptilde_{X}=Q$ in the above
example is to follow the steps of Theorem \ref{thm:EXP_Finite} and Section \ref{sec:EXP_DISC_SIMPLE}
starting with Theorem \ref{thm:EXP_LogBound}.  Since the expectation of the transmitted codeword is outside
the logarithm in \eqref{eq:EXP_JensenPe}, we obtain the constraint 
$\Ptilde_{X}=Q$ in the final minimization using the fact that the empirical distribution
of $\Xv$ is close to $Q$ with high probability.  
We conclude that the inequality in \eqref{eq:EXP_GalEq} is exponentially tight 
for the i.i.d. ensemble when we start with \eqref{eq:EXP_JensenPe}, even  
though it is not tight when we start with \eqref{eq:EXP_GalBound}. 

We have provided two derivations of $\Eexcc$ using the cost-constrained 
ensemble, namely, those in Sections \ref{sec:EXP_DISC_SIMPLE} and 
\ref{sec:EXP_DIST_ENUM} (along with Appendix \ref{sec:EXP_GEN_ANALYSIS}).  
A notable difference between the derivations is
the method for ensuring that the average over $x$ is outside the 
logarithm in \eqref{eq:EXP_Ex_CC}, which is desirable due to Jensen's 
inequality.  In Theorem \ref{thm:EXP_Dual}, the expectation is
inside the logarithm, but the desired result is obtained by choosing  
$a_{2}(x)$ to make Jensen's inequality hold with equality.
On the other hand, in Appendix \ref{sec:EXP_GEN_ANALYSIS} the expectation arises outside the
logarithm even in the case that $L=1$.  

Provided that the assumption of Theorem \ref{thm:EXP_LogBound} is
met, we can combine the two approaches and apply the techniques
of Theorem \ref{thm:EXP_Finite} and Section \ref{sec:EXP_DISC_SIMPLE} to
\eqref{eq:EXP_JensenPe}, in which case $\Excost$ in \eqref{eq:EXP_Ex_Cost}
is improved to
    \begin{equation} 
    \Excoststar(Q,R,\{a_{l}\}) \defeq \sup_{s\ge0,\{\rbar_{l}\}}-\rho\sum_{x}Q(x)\log\sum_{\xbar}Q(\xbar)\Bigg(\sum_{y}W(y|x)\bigg(\frac{q(\xbar,y)}{q(x,y)}\bigg)^{s}e^{\sum_{l=1}^{L}\rbar_{l}(a_{l}(\xbar)-\phi_{l})}\Bigg)^{1/\rho}, \label{eq:EXP_Ex_Cost2}
    \end{equation}
where the outer-most summation arises using Proposition \ref{prop:EXP_LimitQ}
in Appendix \ref{sec:EXP_GEN_ANALYSIS}.
This exponent can also be derived by extending the analysis of Appendix 
\ref{sec:EXP_GEN_ANALYSIS} to include multiple auxiliary costs.

In the case that $L=0$ (i.e. i.i.d. coding), the Lagrange duality techniques
of Theorem \ref{thm:EXP_Duality} reveal that \eqref{eq:EXP_Ex_Cost2} is in fact identical to 
\eqref{eq:EXP_PrimalIID} with the added constraint $P_X=Q$. 
 That is, the additional constraint $P_X=Q$ in the primal expression 
 corresponds to an average over $x$ outside the logarithm in the dual expression.

\subsection{Connections with Statistical Mechanics} \label{sec:EXP_STAT_MECH}

It is instructive to look at the analysis of Sections \ref{sec:EXP_TYPE_ENUM} and \ref{sec:EXP_DIST_ENUM}
from the statistical-mechanical perspective. Let us take another look at the expression
    \begin{equation}
    Z(\xv)=\sum_{\mbar\ne m} e^{-d(\xv,\Xv^{(\mbar)})}, \label{REM}
    \end{equation}
where $d$ can represent either $d_{q}$ (see \eqref{eq:EXP_dq}) or $[d_{s}^{n}]^{+}$ 
(see \eqref{eq:EXP_ChernoffDMulti}).  
From the viewpoint of statistical physics, $Z$ can be interpreted as the
partition function of a physical system, where for a fixed $\xv^{(m)}=\xv$, the
various configurations (microstates) are $\{\xv^{(\mbar)}\}_{\mbar \ne m}$ and
the energy function (Hamiltonian) is given by $d(\xv,\xvbar)$. 
The various ``configurational energies'' $\{d(\xv,\Xv^{(\mbar)})\}$ are 
independent random variables, since the codewords are generated independently. 
As explained in \cite[Ch. 5-6]{PhysicsMezard} (see also \cite[Ch. 6-7]{MerhavPhysics} and
references therein), this setting is analogous to the random energy model (REM) in
the literature of statistical physics of magnetic materials. The REM
was invented by Derrida \cite{Derrida80a,Derrida80b,Derrida81} 
as a model of extremely disordered spin glasses. This model is exactly solvable and exhibits a
phase transition: Below a certain critical temperature, the partition function
becomes dominated by a subexponential number of configurations in the ground-state
energy, which means that the system freezes and its
entropy vanishes in the thermodynamic limit. This combination of freezing and
disorder resembles the behavior of a glass, so this low temperature
phase of zero entropy is called the {\it glassy phase}.
Above the critical temperature, the partition function is dominated by an
exponential number of configurations, so its entropy is positive. This
high temperature phase is called the {\it paramagnetic phase.}

In the case that $P_{\Xv}$ is the constant-composition distribution in \eqref{eq:EXP_Px_CC}
and $d(\cdot,\cdot)$ represents $[d_s^n(\cdot,\cdot)]^{+}$, we can link 
these phases to the exponent $\Eexcc$ in the form given in \eqref{eq:EXP_Eex_RD}.  
The graph of $\Eexcc(Q,R,s)$ is curved at rates below $R_s$ (see \eqref{eq:EXP_Rs}), 
and is a straight line at rates above $R_s$. The curved part
corresponds to the glassy phase of the REM associated with
\eqref{REM}, because the dominant contribution to $\EE[Z(\xv)^{1/\rho}]$ (see \eqref{REM})
is due to a subexponential number of codewords
whose ``distance'' from $\xv$ (i.e. their ``energy'') is roughly $nD_s(Q,R)$. The
straight-line part, on the other hand, corresponds to the paramagnetic phase, where
roughly $e^{n(R-R_s)}$ incorrect codewords at distance $nD_s(Q,R_s)$ dominate the
behavior. Thus, the passage between the curved part and the straight-line
part at $R=R_s$ can be interpreted as a glassy phase transition.  A similar
discussion applies for the multi-letter distance $d_q$ used in 
Section \ref{sec:EXP_TYPE_ENUM}, with $D_{s}(Q,R)$ replaced by
    \begin{equation}
    D_{q}(Q,R) \defeq \min_{\Ptilde_{X\Xbar}\in\SetScc(Q) \,:\, I_{\Ptilde}(X;\Xbar)\le R} D_{q}(\Ptilde_{X\Xbar}),
    \end{equation}
where $D_{q}(\Ptilde_{X\Xbar})$ is defined in \eqref{eq:EXP_DefF}.

\section{Prefactor to the i.i.d. Expurgated Exponent} \label{sec:EXP_PREFACTOR}

Error exponents characterize the rate of decay of the error probability as the
block length increases.  At finite block lengths, the effect of the
subexponential prefactor can be significant, and it is therefore of interest
to characterize its behavior.  There exist several works studying this 
prefactor for the random-coding exponent \cite{TwoChannels,Dobrushin,RefinementJournal,JournalSU}
and the sphere-packing exponent \cite{TwoChannels,Dobrushin,RefinementSP}.
In this section, we characterize the prefactor for the i.i.d. expurgated exponent.
We will see that, under some technical conditions, the prefactor to $\rcux_{\rho}$ 
in \eqref{eq:EXP_RCX} behaves as $O\big(\frac{1}{\sqrt n}\big)$, thus improving
on Gallager's $O(1)$ prefactor. Our analysis builds on that of \cite{JournalSU,PaperRefinement}. 

\subsection{Preliminary Definitions} \label{sec:EXP_PRELIMINARY_LEMMAS}

We define the sets 
    \begin{align}
    \Yc_1(x,\xbar) & \defeq \Big\{ y \,:\, W(y|x)W(y|\xbar)>0  \Big\} \label{eq:EXP_SetY1} \\
    \Ac(Q) & \defeq \bigg\{ (x,\xbar) \,:\, Q(x)Q(\xbar)>0, \, \frac{q(\xbar,y)}{q(x,y)} \ne \frac{q(\xbar,y^\prime)}{q(x,y^\prime)} \text{ for some } y,y^\prime \in \Yc_1(x,\xbar) \bigg\} \label{eq:EXP_SetA}
    \end{align}
and make the following technical assumptions:
    \begin{equation}
    q(x,y) = 0 \iff W(y|x) = 0 \label{eq:EXP_Assumption1}
    \end{equation}
    \begin{equation}
    \Ac(Q) \ne \emptyset. \label{eq:EXP_Assumption2}
    \end{equation}
In the case that $q(x,y)=W(y|x)$ (i.e. ML decoding), \eqref{eq:EXP_Assumption1} 
is trivial, and \eqref{eq:EXP_Assumption2} reduces to the \emph{non-singularity}
assumption of \cite{RefinementJournal}.  A notable example where this condition
fails is the binary erasure channel (BEC) with $Q=\big(\frac{1}{2},\frac{1}{2}\big)$.

We write
    \begin{equation}
    \Exiid(Q,\rho,s) \defeq -\rho\log\sum_{x,\xbar}Q(x)Q(\xbar)\Bigg(\sum_{y}W(y|x)\bigg(\frac{q(\xbar,y)}{q(x,y)}\bigg)^{s}\Bigg)^{1/\rho} \label{eq:EXP_Ex_IID_s}
    \end{equation} 
to denote the objective in \eqref{eq:EXP_Ex_IID} with a fixed value of $s$.
We define the tiled distribution
    \begin{align}
        V_{s}(y|x,\xbar) & \defeq \frac{W(y|x)\Big(\frac{q(\xbar,y)}{q(x,y)}\Big)^s}{\sum_{y^\prime}W(y^\prime|x)\Big(\frac{q(\xbar,y^\prime)}{q(x,y^\prime)}\Big)^s} \\
        V_{s}^{n}(\yv|\xv,\xvbar) & \defeq \prod_{i=1}^{n}V_{s}(y_i|x_i,\xbar_i),
    \end{align}
and the generalized information density
    \begin{align}
        j_{s}(x,\xbar,y) &\defeq \log\frac{V_{s}(y|x,\xbar)}{W(y|x)} \label{eq:EXP_jSingle} \\
        j_{s}^{n}(\xv,\xvbar,\yv) &\defeq \sum_{i=1}^{n}j_{s}(x_i,\xbar_i,y_i). \label{eq:EXP_jMulti} 
    \end{align}
Furthermore, we define the joint tilted distribution
    \begin{equation}
    P_{\rho,s}^{*}(x,\xbar) = \frac{Q(x)Q(\xbar)\Big(\sum_{y}W(y|x)\Big(\frac{q(\xbar,y)}{q(x,y)}\Big)^s\Big)^{1/\rho}}{\sum_{x^\prime,\xbar^\prime}Q(x^\prime)Q(\xbar^\prime)\Big(\sum_{y^\prime}W(y^\prime|x^\prime)\Big(\frac{q(\xbar^\prime,y^\prime)}{q(x^\prime,y^\prime)}\Big)^s\Big)^{1/\rho}}, \label{eq:EXP_PXX*}
    \end{equation}
and the conditional variance
    \begin{equation}
    c_0(Q,\rho,s) \triangleq \EE\Big[ \var\big[j_s(X^{*}_s,\Xbar^{*}_s,Y_s^{*}) \big| X^{*}_s,\Xbar^{*}_s \big] \Big], \label{eq:EXP_c0}
    \end{equation}
where $(X_s^*,\Xbar_s^*,Y_s^*) \sim P_{\rho,s}^{*}(x,\xbar)V_s(y|x,\xbar)$.
The arguments to $c_0$ will henceforth be omitted, since their values will be 
understood from the context.

Writing $Y_s \sim V_s(\cdot|x,\xbar)$, the following arguments show that the assumptions in 
\eqref{eq:EXP_Assumption1}--\eqref{eq:EXP_Assumption2} imply that $c_0>0$ whenever $s>0$:
    \begin{align}
    \var[j_{s}(x,\xbar,Y_{s})]=0 & \iff j_{s}(x,\xbar,y)\text{ is independent of } y \text{ wherever }V_{s}(y|x,\xbar)>0 \label{eq:EXP_Var0_0} \\
     & \iff \frac{q(\xbar,y)}{q(x,y)}\text{ is independent of } y \text{ wherever }W(y|x)q(\xbar,y)>0 \label{eq:EXP_Var0_1} \\
     & \iff (x,\xbar) \notin \Ac(Q), \label{eq:EXP_Var0_2}
    \end{align}
where \eqref{eq:EXP_Var0_1} follows from the definitions of $j_s$ and $V_s$, and \eqref{eq:EXP_Var0_2} follows
from the assumption in \eqref{eq:EXP_Assumption1} and the definition of $\Ac(Q)$. Using the assumption 
in \eqref{eq:EXP_Assumption2}, it follows that $c_0>0$.

Finally, we define the set
\begin{equation}
    \Ic_s \triangleq \Big\{ j_s(x,x,y) \,:\, W(y|x)>0, (x,\xbar)\in\Ac(Q) \Big\}
\end{equation}
and the constant
\begin{equation}
    \psi_{s} \triangleq 
        \begin{cases}
            1 & \Ic_s\text{ does not lie on a lattice} \\
            \frac{\hover}{1-e^{-\hover}} & \Ic_s\text{ lies on a lattice with span }\hover.
        \end{cases} \label{eq:SA_Psi_s}
\end{equation}

\subsection{Statement of the Result}

\begin{thm} \label{thm:EXP_Prefactor}
    Fix any DMC $W$, decoding metric $q$ and input distribution $Q$ satisfying 
    \eqref{eq:EXP_Assumption1}--\eqref{eq:EXP_Assumption2}.  For any $R>0$, $\rho\ge1$ and
    $s>0$, there exists a sequence of codebooks $\Cc_n$ with $M \ge e^{nR}$ codewords 
    whose maximal error probability satisfies
        \begin{equation}
        p_e(\Cc_n) \le \frac{4^\rho \psi_s}{\sqrt{2\pi nc_0}}e^{-n(\Exiid(Q,\rho,s)-\rho R)}\big(1+o(1)\big) \label{eq:EXP_Prefactor}
        \end{equation}
\end{thm}
\begin{proof}
    See Section \ref{sec:EXP_PREFACTOR_PROOF}.
\end{proof}

It is interesting to note that under ML coding and any rate where the expurgated 
exponent and random-coding exponent coincide (i.e. $\rho=1$ in both cases), 
Theorem \ref{thm:EXP_Prefactor} gives the same prefactor growth rate as 
that of the random-coding exponent \cite{RefinementJournal,JournalSU}.  
There is an extra factor of four in \eqref{eq:EXP_Prefactor},
which can be attributed to the fact that Theorem \ref{thm:EXP_Prefactor} considers
the maximal error rather than the average error. Of course, Theorem \ref{thm:EXP_Prefactor} is
primarily of interest at low rates, where the expurgated exponent exceeds the random-coding exponent.

\subsection{Proof of Theorem \ref{thm:EXP_Prefactor}} \label{sec:EXP_PREFACTOR_PROOF}

The proof makes use of two technical lemmas.  The first is a strong
large deviations result which was proved in \cite{JournalSU}, building
upon the analysis in the proof of \cite[Lemma 47]{Finite}.

\begin{lem} \label{lem:REF_Lem20}
    {\em \cite[Lemma 1]{JournalSU}}
    Fix $K>0$, and for each $n$, let $(n_1,\cdots,n_K)$ be integers such that $\sum_{k}n_k=n$.
    Fix the PMFs $Q_1,\cdots,Q_K$ on
    a finite subset of $\RR$, and let $\sigma_1^2,\cdots,\sigma_K^2$ be the corresponding
    variances.  Let $Z_{1},\cdots,Z_{n}$ be independent random variables, $n_k$ of which
    are distributed according to $Q_k$ for each $k$.  Suppose that 
    $\min_{k}\sigma_k > 0$ and $\min_{k}n_k = \Theta(n)$.  Defining
    \begin{align}
        \Ic_0  &\triangleq \bigcup_{k \,:\, \sigma_{k} > 0}\big\{ z \,:\, Q_k(z)>0 \big\} \label{eq:SA_I0} \\
        \psi_0 &\triangleq 
            \begin{cases}
                1 & \Ic_0\text{ \em does not lie on a lattice} \\
                \frac{h_0}{1-e^{-h_0}} & \Ic_0\text{ \em lies on a lattice with span }h_0,
            \end{cases} \label{eq:SA_Psi_0}
    \end{align}
    the summation $S_n\triangleq\sum_{i}Z_i$ satisfies the 
    following uniformly in $t$:
    \begin{equation}
        \EE\Big[e^{-S_n}\emph{\openone}\big\{S_n \ge t\big\}\Big] \le e^{-t}\bigg(\frac{\psi_0}{\sqrt{2\pi V_n}} + o\Big(\frac{1}{\sqrt{n}}\Big)\bigg), \label{eq:SU_Lemma20}
    \end{equation}
    where $V_n \triangleq \var[S_n]$.
\end{lem}

The following lemma ensures the existence of a high probability set of $(\xv,\xvbar)$
pairs such that Lemma \ref{lem:REF_Lem20} can be applied to the inner probability in \eqref{eq:EXP_RCX}. 

\begin{lem} \label{lem:EXP_TechnicalLemma}
    For any $R>0$, $\rho\ge1$, $s>0$ and $(W,q,Q)$ satisfying 
    \eqref{eq:EXP_Assumption1}--\eqref{eq:EXP_Assumption2}, the sequence of sets 
        \begin{equation}
            \Fc^n_{\rho,s}(\delta) \defeq \Big\{ (\xv,\xvbar) \,:\, \max_{x,\xbar} \Big| \hat{P}_{\xv \xvbar}(x,\xbar) - P^{*}_{\rho,s}(x,\xbar)\Big| \le \delta \Big\} \label{eq:EXP_SetFn}
        \end{equation}
    satisfies the following properties:
        \begin{enumerate}
            \item For any $\delta>0$ and $(\xv,\xvbar) \in \Fc^n_{\rho,s}(\delta)$, the random variable 
                  $\Yv_s \sim V_{s}^{n}(\cdot|\xv,\xvbar)$ satisfies
                      \begin{equation}
                      \var[j_{s}^{n}(\xv,\xvbar,\Yv_s)] \ge n(c_0 - r(\delta)), \label{eq:EXP_VarLB}
                      \end{equation}
                  where $r(\delta)\to 0$ as $\delta\to 0$.
            \item For any $\delta>0$, we have
                  \begin{equation}
                      \liminf_{n\to\infty} -\frac{1}{n} \frac{ \sum_{(\xv,\xvbar) \notin \Fc^n_{\rho,s}(\delta)} Q^n(\xv)Q^n(\xvbar)\Big( \sum_{\yv}W^n(\yv|\xv)\Big(\frac{q(\xvbar,\yv)}{q(\xv,\yv)}\Big)^s\Big)^{1/\rho} }{\sum_{\xv,\xvbar} Q^n(\xv)Q^n(\xvbar)\Big( \sum_{\yv}W^n(\yv|\xv)\Big(\frac{q(\xvbar,\yv)}{q(\xv,\yv)}\Big)^s\Big)^{1/\rho}} > 0. \label{eq:EXP_Ratio} 
                  \end{equation}
        \end{enumerate}
\end{lem}

\begin{proof}[Proof of Theorem \ref{thm:EXP_Prefactor} Based on Lemma \ref{lem:EXP_TechnicalLemma}]
    Using the bound $\rcux_{\rho}$ in Theorem \ref{thm:EXP_Finite} with 
    the i.i.d. codeword distribution $P_{\Xv}=Q^n$, we have for any $\delta > 0$ that
        \begin{align}
        \frac{1}{4(M-1)} \rcux_{\rho}(n,M)^{1/\rho} & = \sum_{\xv,\xvbar} Q^n(\xv)Q^n(\xvbar)\PP\Big[q^{n}(\xvbar,\Yv)\ge q^{n}(\xv,\Yv)\Big]^{1/\rho} \\
        & = \sum_{(\xv,\xvbar) \in \Fc^n_{\rho,s}(\delta)} Q^n(\xv)Q^n(\xvbar)\PP\Big[q^{n}(\xvbar,\Yv)\ge q^{n}(\xv,\Yv)\Big]^{1/\rho} \nonumber \\
        & \qquad + \sum_{(\xv,\xvbar) \notin \Fc^n_{\rho,s}(\delta)} Q^n(\xv)Q^n(\xvbar)\PP\Big[q^{n}(\xvbar,\Yv)\ge q^{n}(\xv,\Yv)\Big]^{1/\rho}, \label{eq:EXP_PreProof2}
        \end{align}
    where each probability is implicitly conditioned on $\Xv=\xv$. 
    
    We first analyze the summation over $\Fc^n_{\rho,s}(\delta)$ in \eqref{eq:EXP_PreProof2}. 
    In order to make the inner probability more  amenable to an application of 
    Lemma \ref{lem:REF_Lem20}, we write it as
        \begin{align}
        \PP\Big[q^n(\xvbar,\Yv)\ge q^n(\xv,\Yv)\Big] & =\PP\bigg[\bigg(\frac{q^n(\xvbar,\Yv)}{q^n(\xv,\Yv)}\bigg)^{s}\ge1\bigg]\\
         & =\PP\left[\frac{\Big(\frac{q^n(\xvbar,\Yv)}{q^n(\xv,\Yv)}\Big)^{s}}{\sum_{\yv}W^{n}(\yv|\xv)\Big(\frac{q^n(\xvbar,\yv)}{q^n(\xv,\yv)}\Big)^{s}}\ge\frac{1}{\sum_{\yv}W^{n}(\yv|\xv)\Big(\frac{q^n(\xvbar,\yv)}{q^n(\xv,\yv)}\Big)^{s}}\right]\\
         & =\PP\left[j_{s}^{n}(\xv,\xvbar,\Yv)\ge-\log\sum_{\yv}W^{n}(\yv|\xv)\bigg(\frac{q(\xvbar,\yv)}{q(\xv,\yv)}\bigg)^{s}\right], \label{eq:EXP_PreProof6}
        \end{align}
    where $j_{s}^n$ is defined in \eqref{eq:EXP_jMulti}.  Next, following \cite[Sec. 3.4.5]{FiniteThesis}, we note
    that the following holds when $V_{s}^{n}(\yv|\xv,\xvbar) \ne 0$:
        \begin{align}
        W^n(\yv|\xv) & = W^n(\yv|\xv)\frac{V_{s}^{n}(\yv|\xv,\xvbar)}{V_{s}^{n}(\yv|\xv,\xvbar)} \\
        & = V_{s}^{n}(\yv|\xv,\xvbar)e^{-nj_{s}(\xv,\xvbar,\yv)}. \label{eq:EXP_PreProof7}
        \end{align}
    Summing \eqref{eq:EXP_PreProof7} over all $\yv$ such that $j_{s}(\xv,\xvbar,\yv) \ge t$, we obtain
        \begin{equation}
        \PP\big[j_{s}^{n}(\xv,\xvbar,\Yv)\ge t\big]=\EE\Big[e^{-j_{s}^{n}(\xv,\xvbar,\Yv_{s})}\openone\big\{ j_{s}^{n}(\xv,\xvbar,\Yv_{s})\ge t\big\}\Big],
        \end{equation}
    where $\Yv_{s} \sim V_{s}^{n}(\cdot|\xv,\xvbar)$.  For any $(\xv,\xvbar) \in \Fc^n_{\rho,s}(\delta)$,
    we obtain the following using Lemma \ref{lem:REF_Lem20}, 
    the first part of Lemma \ref{lem:EXP_TechnicalLemma}, and the fact that $c_0>0$ (see the arguments
    following \eqref{eq:EXP_Var0_0}):
        \begin{equation}
        \PP\big[j_{s}^{n}(\xv,\xvbar,\Yv)\ge t\big] \le \frac{\psi_s(1+o(1))}{\sqrt{2\pi n(c_0 - r(\delta))}}e^{-t} \label{eq:EXP_PreProof9}
        \end{equation}
    uniformly in $t$, provided that $\delta$ is sufficiently small so that $r(\delta)<c_0$.  
    Substituting \eqref{eq:EXP_PreProof9} into \eqref{eq:EXP_PreProof6}, we obtain
        \begin{equation}
        \PP\Big[q^n(\xvbar,\Yv)\ge q^n(\xv,\Yv)\Big] \le \frac{\psi_s(1+o(1))}{\sqrt{2\pi n(c_0 - r(\delta))}}\sum_{\yv} W^{n}(\yv|\xv)\bigg(\frac{q(\xvbar,\yv)}{q(\xv,\yv)}\bigg)^{s},
        \end{equation}
    and hence
        \begin{align}
        & \sum_{(\xv,\xvbar) \in \Fc^n_{\rho,s}(\delta)} Q^n(\xv)Q^n(\xvbar)\PP\Big[q^{n}(\xvbar,\Yv)\ge q^{n}(\xv,\Yv)\Big]^{1/\rho} \nonumber \\ 
        & \qquad \le \sum_{\xv,\xvbar} Q^n(\xv)Q^n(\xvbar) \bigg(\frac{\psi_s(1+o(1))}{\sqrt{2\pi n(c_0 - r(\delta))}} \sum_{\yv} W^{n}(\yv|\xv)\bigg(\frac{q(\xvbar,\yv)}{q(\xv,\yv)} \bigg)^{s} \bigg)^{1/\rho}. \label{eq:EXP_PreProof12}
        \end{align}
    We observe that the right-hand side of \eqref{eq:EXP_PreProof12} has the same exponent as 
    the denominator of \eqref{eq:EXP_Ratio}.  Using Markov's inequality, the summation over 
    $\Fc^n_{\rho,s}(\delta)^{c}$ in \eqref{eq:EXP_PreProof2} can be upper bounded by the numerator of
    \eqref{eq:EXP_Ratio}, and thus the second part of Lemma \ref{lem:EXP_TechnicalLemma} implies
        \begin{equation}
        \frac{1}{4(M-1)} \rcux_{\rho,s}(n,M)^{1/\rho} \le \big(1+o(1)\big) \sum_{\xv,\xvbar} Q^n(\xv)Q^n(\xvbar) \bigg(\frac{\psi_s(1+o(1))}{\sqrt{2\pi n(c_0 - r(\delta))}} \sum_{\yv} W^{n}(\yv|\xv)\bigg(\frac{q(\xvbar,\yv)}{q(\xv,\yv)} \bigg)^{s} \bigg)^{1/\rho},
        \end{equation}
    and hence
        \begin{align}
        \rcux_{\rho,s}(n,M) & \le \frac{4^{\rho}\psi_s(1+o(1))}{\sqrt{2\pi n(c_0 - r(\delta))}} M^\rho \Bigg(\sum_{\xv,\xvbar} Q^n(\xv)Q^n(\xvbar) \bigg(\sum_{\yv} W^{n}(\yv|\xv)\bigg(\frac{q(\xvbar,\yv)}{q(\xv,\yv)} \bigg)^{s} \bigg)^{1/\rho}\Bigg)^\rho \\
        & = \frac{4^{\rho}\psi_s(1+o(1))}{\sqrt{2\pi n(c_0 - r(\delta))}} e^{-n(\Exiid(Q,\rho,s)-\rho R)}, \label{eq:EXP_PreProof15}
        \end{align}
    where \eqref{eq:EXP_PreProof15} follows by expanding each term as a product 
    from $1$ to $n$ and using the definition of $\Exiid$.  The proof is concluded by taking
    $\delta\to0$ (and hence $r(\delta)\to0$).
\end{proof}

\begin{proof}[Proof of Lemma \ref{lem:EXP_TechnicalLemma}]
    We obtain \eqref{eq:EXP_VarLB} by expanding the variance as
        \begin{align}
        \var[j_{s}^{n}(\xv,\xvbar,\Yv_s)] & = \sum_{i=1}^{n} \var[j_s(x_i,\xbar_i,Y_{s,i})] \\
        & = \sum_{x,\xbar} n\hat{P}_{\xv \xvbar}(x,\xbar) \var[j_s(x,\xbar,Y_{s})]
        \end{align}
    and substituting the bound in the definition of $\Fc^n_{\rho,s}(\delta)$ in \eqref{eq:EXP_SetFn}.  
    To prove the second property, we note that a nearly identical argument 
    to Section \ref{sec:EXP_DISC_SIMPLE} (based on types) reveals that the exponent of 
    the denominator of \eqref{eq:EXP_Ratio} is equal to
        \begin{equation}
        \min_{P_{X\Xbar}} D(P_{X\Xbar} \| Q \times Q) + \frac{1}{\rho} \EE_{P}[d_s(X,\Xbar)], \label{eq:EXP_Quantity1Exp}
        \end{equation}
    where $d_s$ is defined in \eqref{eq:EXP_ChernoffDist}.
    Similarly, the  exponent of the numerator of \eqref{eq:EXP_Ratio} is given by
        \begin{equation}
        \min_{P_{X\Xbar} \,:\, \max_{x,\xbar}|P_{X\Xbar}(x,\xbar)-P_{\rho,s}^{*}(x,\xbar)|>\delta} D(P_{X\Xbar} \| Q \times Q) + \frac{1}{\rho} \EE_{P}[d_s(X,\Xbar)]. \label{eq:EXP_Quantity2Exp}
        \end{equation}
    A straightforward analysis of the Karush-Kuhn-Tucker (KKT) conditions \cite[Sec. 5.5.3]{Convex} reveals  
    that \eqref{eq:EXP_Quantity1Exp} is uniquely minimized by $P_{\rho,s}^{*}$, defined in \eqref{eq:EXP_PXX*}.
    On the other hand, $P_{\rho,s}^{*}$ does not satisfy the constraint in in \eqref{eq:EXP_Quantity2Exp},
    and thus \eqref{eq:EXP_Quantity2Exp} is strictly higher than \eqref{eq:EXP_Quantity1Exp}.
\end{proof}

\section{Discussion and Conclusion} \label{sec:EXP_CONCLUSION}

We have presented asymptotic and non-asymptotic expurgated bounds 
for channels with a given decoding rule.  Several expurgated exponents
have been derived, including that of Csisz\'{a}r and K\"{o}rner \cite{Csiszar1} and its generalization to 
continuous alphabets.  The type class enumeration approach has been shown to 
provide better exponents for some codeword distributions, better guarantees of exponential 
tightness, and the opportunity for deriving expurgated exponents for channels with memory.
By refining the analysis of the i.i.d. ensemble, we have obtained a bound with a 
$O\big(\frac{1}{\sqrt{n}}\big)$ prefactor, thus improving on Gallager's $O(1)$ prefactor.

\appendix
\section{Appendix}

\subsection{Technical Condition of Theorem \ref{thm:EXP_LogBound}} \label{sec:EXP_TECH_COND}

We begin by providing an example of a class of continuous channels and metrics satisfying the
single-letter condition given in \eqref{eq:EXP_SingleLetterCond}.  Consider an additive noise channel $Y=X+Z$,
and let $q(x,y)$ be any decreasing function of $|y-x|$. 
If the cost constraint is of the form $c(x)=|x|^{\beta}$ for some constant $\beta$,
then $c(x) \le \gamma$ if and only if $|x|\le\gamma^{1/\beta}$.  Thus,
any two permissible points are separated by a distance of at most $2\gamma^{1/\beta}$,
and the single-letter condition is satisfied if the additive noise satisfies
$\PP[Z > 2\gamma^{1/\beta}] \ge e^{-E'(\gamma)}$ and $\PP[Z < -2\gamma^{1/\beta}] \ge e^{-E'(\gamma)}$
for some $E'(\gamma)$ growing subexponentially in $\gamma$.  In particular, this holds for 
noise distributions with exponential tails (e.g. Gaussian).
On the other hand, if the cost function is logarithmic, say $c(x) = \log(1+|x|)$, then
\eqref{eq:EXP_SingleLetterCond} fails for additive noise distributions with exponential 
tails, since in this case the limit on the left-hand side of \eqref{eq:EXP_SingleLetterCond} 
equals a positive constant.

For any DMC whose zero-error capacity \cite{ShannonZero} is zero, the condition of
Theorem \ref{thm:EXP_LogBound} is satisfied under ML decoding, since the error probability can only
decay exponentially.  On the other hand, the condition could
fail for  sufficiently ``bad'' metrics (e.g. one for which there exists a pair $(x,\xbar)$   
such that $q(x,y)>q(\xbar,y)$ for all $y$).  Furthermore, the condition fails
under ML decoding whenever the zero-error capacity is positive and $Q$ 
has a support which includes two inputs not sharing a common output. 

\subsection{Proof of Theorem \ref{thm:EXP_Duality}} \label{sec:EXP_DUALITY_PROOF}

Using the definitions of $\SetScc$ and $\SetTcc$ in \eqref{eq:EXP_SetS}--\eqref{eq:EXP_SetT2},
we write \eqref{eq:EXP_PrimalAlt} as 
    \begin{equation}
    \Eexhatcc(Q,R)=\min_{\substack{\Ptilde_{X\Xbar}\in\SetScc(Q)\\I_{\Ptilde}(X;\Xbar)\le R}}\min_{P_{X \Xbar Y}\in\SetTcc(\Ptilde_{X\Xbar})}D(P_{X \Xbar Y}\|\Ptilde_{X\Xbar}\times W) + I_{\Ptilde}(X;\Xbar) - R,\label{eq:EXP_DualProof1}
    \end{equation}
where the objective follows from \eqref{eq:EXP_ChainRule}.  We will study
\eqref{eq:EXP_DualProof1} one minimization at a time.  

\subsubsection*{Step 1}

For a given $\Ptilde_{X\Xbar}\in\SetScc(Q)$, $I_{\Ptilde}(X;\Xbar)-R$ is constant, 
and we thus consider the optimization problem
    \begin{equation}
    \min_{P_{X \Xbar Y} \in \SetTcc(\Ptilde_{X\Xbar})}D(P_{X \Xbar Y}\|\Ptilde_{X\Xbar}\times W). \label{eq:EXP_DualProof2}
    \end{equation}
The Lagrangian \cite[Sec. 5.1.1]{Convex} is given by
    \begin{multline}
    \Lsf_{1}=\sum_{x,\xbar,y}P_{X \Xbar Y}(x,\xbar,y)\log\frac{P_{X \Xbar Y}(x,\xbar,y)}{\Ptilde_{X\Xbar}(x,\xbar)W(y|x)}\\
    +s\bigg(\sum_{x,y}P_{XY}(x,y)\log q(x,y)-\sum_{\xbar,y}P_{\Xbar Y}(\xbar,y)\log q(\xbar,y)\bigg)+\sum_{x,\xbar}\mu(x,\xbar)\Big(\Ptilde_{X\Xbar}(x,\xbar)-P_{X\Xbar}(x,\xbar)\Big),\label{eq:EXP_DualProof3}
    \end{multline}
where $s\ge0$ and $\mu(\cdot,\cdot)$ are Lagrange multipliers.
The optimization problem is convex with affine constraints, and thus the optimal value is equal 
to $\Lsf_{1}$ for some choice of $P_{X \Xbar Y}$ and the Lagrange
multipliers satisfying the Karush-Kuhn-Tucker (KKT) conditions \cite[Sec. 5.5.3]{Convex}.

The simplification of \eqref{eq:EXP_DualProof3} using the KKT conditions 
uses standard arguments, so we omit some details.
Setting $\frac{\partial \Lsf_{1}}{\partial P_{X \Xbar Y}(x,\xbar,y)}=0$,
using the constraint $P_{X\Xbar}=\Ptilde_{X\Xbar}$
to solve for $\mu(\cdot,\cdot)$, and substituting the resulting expressions back
into \eqref{eq:EXP_DualProof3}, we obtain
    \begin{equation}
    \Lsf_{1} = -\sum_{x,\xbar}\Ptilde_{X\Xbar}(x,\xbar)\log\sum_{y}W(y|x)\bigg(\frac{q(\xbar,y)}{q(x,y)}\bigg)^{s}. \label{eq:EXP_DualProof4}
    \end{equation}
Renaming $\Ptilde_{X\Xbar}$ as $P_{X\Xbar}$, taking the supremum over $s\ge0$, and
adding $I_{P}(X;\Xbar)-R$ (see \eqref{eq:EXP_DualProof1}--\eqref{eq:EXP_DualProof2}), 
we obtain the right-hand side of \eqref{eq:EXP_Duality1} with the 
minimum and supremum in the opposite order. Using Fan's minimax theorem 
\cite{Minimax}, we can safely interchange the two.

Since we have taken the supremum over the parameter $s\ge0$ without 
verifying that it satisfies the KKT conditions,
we have only proved that \eqref{eq:EXP_Duality1} holds with the equality replaced
by an inequality ($\le$).  To prove the reverse inequality, we use
the log-sum inequality \cite[Thm. 2.7.1]{Cover} similarly to \cite[Appendix A]{Merhav}. 
For any $P_{X \Xbar Y}\in\SetTcc(\Ptilde_{X\Xbar})$, we have
\begin{align}
    D(P_{X \Xbar Y}\|\Ptilde_{X\Xbar}\times W) 
        &\ge D(P_{X \Xbar Y}\|\Ptilde_{X\Xbar}\times W) - s\sum_{x,\xbar,y} P_{X \Xbar Y}(x,\xbar,y)\log\frac{q(\xbar,y)}{q(x,y)} \label{eq:EXP_DualLB1} \\
        &= \sum_{x,\xbar,y} P_{X \Xbar Y}(x,\xbar,y)\log\frac{P_{X \Xbar Y}(x,\xbar,y)}{\Ptilde_{X\Xbar}(x,\xbar)W(y|x)\Big(\frac{q(\xbar,y)}{q(x,y)}\Big)^{s}} \label{eq:EXP_DualLB2} \\
        &\ge \sum_{x,\xbar} P_{X \Xbar}(x,\xbar)\log\frac{1}{\sum_{y}W(y|x)\Big(\frac{q(\xbar,y)}{q(x,y)}\Big)^{s}} \label{eq:EXP_DualLB3},
\end{align}
where \eqref{eq:EXP_DualLB1} holds for any $s\ge0$ from the constraint $\EE_{P}[\log q(\Xbar,Y)]\ge\EE_{P}[\log q(X,Y)]$
in \eqref{eq:EXP_SetT}, \eqref{eq:EXP_DualLB2} follows from the definition of divergence, and \eqref{eq:EXP_DualLB3}
follows using the log-sum inequality \cite[Thm. 2.7.1]{Cover} and the constraint $P_{X\Xbar}=\Ptilde_{X\Xbar}$.
Equation \eqref{eq:EXP_DualLB3} coincides with \eqref{eq:EXP_DualProof4}, thus
completing the proof of \eqref{eq:EXP_Duality1}.

\subsubsection*{Step 2}

We now turn to the proof of \eqref{eq:EXP_Eex_CC}.
For any fixed $s\ge0$, the Lagrangian corresponding to \eqref{eq:EXP_Duality1} is given by 
    \begin{multline}
    \Lsf_{2} = -\sum_{x,\xbar}P_{X\Xbar}(x,\xbar)\log\sum_{y}W(y|x)\bigg(\frac{q(\xbar,y)}{q(x,y)}\bigg)^{s}+(1+\lambda)\sum_{x,\xbar}P_{X\Xbar}(x,\xbar)\log\frac{P_{X\Xbar}(x,\xbar)}{Q(x)Q(\xbar)} - (1+\lambda)R \\
    +\sum_{x}\nu_{1}(x)\Big(Q(x)-P_{X}(x)\Big)+\sum_{\xbar}\nu_{2}(\xbar)\Big(Q(\xbar)-P_{\Xbar}(\xbar)\Big),\label{eq:EXP_DualProofB2}
    \end{multline}
where $\lambda\ge0$,  $\nu_{1}(\cdot)$ and $\nu_{2}(\cdot)$ are Lagrange multipliers.
Setting $\frac{\partial \Lsf_{2}}{\partial P_{X\Xbar}(x,\xbar)}=0$, using the constraint 
$P_{X}=Q$ to solve for $\nu_{1}(\cdot)$, and substituting the resulting expressions back
into \eqref{eq:EXP_DualProofB2}, we obtain
    \begin{equation}
    \Lsf_{2} =-(1+\lambda)\sum_{x}Q(x)\log\sum_{\xbar}Q(\xbar)\bigg(\sum_{y}W(y|x)\bigg(\frac{q(\xbar,y)}{q(x,y)}\bigg)^{s}\bigg)^{\frac{1}{1+\lambda}}e^{\frac{1}{1+\lambda}(\nu_{2}(\xbar)-\nu_{2}(x))} - (1+\lambda)R. \label{eq:EXP_DualProofB8}
    \end{equation}
Taking the supremum over $\nu_{2}(\cdot)$, $s\ge0$ and $\lambda \ge 0$, we obtain
the right-hand side of \eqref{eq:EXP_Eex_CC} after suitable renaming. 

Once again, we have only proved that \eqref{eq:EXP_Eex_CC} holds with an inequality ($\le$)
in place of the equality, and we obtain a matching lower bound similarly to \eqref{eq:EXP_DualLB1}--\eqref{eq:EXP_DualLB3}.  
For any $P_{X\Xbar}\in\SetScc(Q)$ with $I_{\Ptilde}(X;\Xbar)\le R$, 
we can lower bound the objective in \eqref{eq:EXP_Duality1} as follows: 
    \begin{align}
     -\sum_{x,\xbar} &P_{X\Xbar}(x,\xbar)\log\sum_{y}W(y|x)\bigg(\frac{q(\xbar,y)}{q(x,y)}\bigg)^{s} + I_{P}(X;\Xbar)-R \nonumber \\
    & \ge -\sum_{x,\xbar}P_{X\Xbar}(x,\xbar)\log\sum_{y}W(y|x)\bigg(\frac{q(\xbar,y)}{q(x,y)}\bigg)^{s} + \rho\big(I_{P}(X;\Xbar)-R\big) \label{eq:EXP_DualProofB10} \\
    &  = -\rho\sum_{x,\xbar}P_{X\Xbar}(x,\xbar)\log\frac{Q(x)Q(\xbar)\bigg(\sum_{y}W(y|x)\Big(\frac{q(\xbar,y)}{q(x,y)}\Big)^{s}e^{a(\xbar)-\phi_a}\bigg)^{1/\rho}}{P_{X\Xbar}(x,\xbar)} - \rho R \label{eq:EXP_DualProofB11} \\
    & \ge -\rho\sum_{x}Q(x)\log\sum_{\xbar}Q(\xbar)\bigg(\sum_{y}W(y|x)\Big(\frac{q(\xbar,y)}{q(x,y)}\Big)^{s}e^{a(\xbar)-\phi_a}\bigg)^{1/\rho}  - \rho R, \label{eq:EXP_DualProofB13}
    \end{align}
where \eqref{eq:EXP_DualProofB10} holds for any $\rho\ge1$ from the constraint $I_{\Ptilde}(X;\Xbar) \le R$,
\eqref{eq:EXP_DualProofB11} holds for any function $a(x)$ with mean $\phi_a = \EE_Q[a(X)]$ by expanding 
the logarithm and applying simple manipulations, and \eqref{eq:EXP_DualProofB13} follows from 
the log-sum inequality \cite[Thm. 2.7.1]{Cover} and the constraint $P_X=Q$.
Using the definition of $\phi_a$ and again expanding the logarithm, it is easily shown that
\eqref{eq:EXP_DualProofB13} is unchanged when $e^{a(\xbar)-\phi_a}$ is replaced by
$\frac{e^{a(\xbar)}}{e^{a(x)}}$, thus completing the proof.

\subsection{Proof of Proposition \ref{prop:EXP_RateZero}} \label{sec:EXP_RZERO_PROOF}

The result for the i.i.d. exponent follows similarly to Gallager \cite[Sec 5.7]{Gallager},
so we only explain the differences. Let $\Exiid(Q,\rho,s)$ be the function 
$\Exiid$ in \eqref{eq:EXP_Ex_IID}, with a fixed value of $s$ rather than a supremum.  
We claim that
    \begin{equation}
    \lim_{R\rightarrow0^{+}}\sup_{\rho\geq1,s\ge0}\Exiid(Q,\rho,s)-\rho R
    = \sup_{\rho\geq1,s\ge0}\Exiid(Q,\rho,s). \label{eq:EXP_SupInf1}
    \end{equation}
It is easily seen that the left-hand side of \eqref{eq:EXP_SupInf1} cannot exceed the right-hand
side, since $\rho R$ is positive for any sequence of $R$ values approaching zero from above.  It 
remains to prove the converse. We have for all $R$ that
    \begin{equation}
    \sup_{\rho\ge1,s\ge0} \Exiid(Q,\rho,s) - \rho R \ge \Exiid(Q,\rho,s)-\rho R.
    \end{equation}
Taking $R \to 0$ and then taking the supremum over $s\ge0$ and $\rho\ge1$ yields the desired result.
The remainder of the proof follows using Gallager's argument: For any fixed $s$, the supremum over 
$\rho$ is in the limit as $\rho\to\infty$, and this limit is easily evaluated using L'H\^{o}pital's rule.

The result for the constant-composition exponent follows in the same way 
using the fact that \linebreak $\sup_{s,a_1(\cdot),a_2(\cdot)}\Excost(Q,\rho,\{a_1,a_2\}) = \Excc(Q,\rho)$ 
(see Section \ref{sec:EXP_DISC_SIMPLE}; in particular, $\Excost$ is defined in \eqref{eq:EXP_Ex_Cost}).
Once again, the supremum over $\rho$ is in the limit as $\rho\to\infty$ when the 
remaining parameters are fixed.

\subsection{Derivation of $\Eexcc$ Using Theorem \ref{thm:EXP_General}} \label{sec:EXP_GEN_ANALYSIS}

Using similar arguments to those in Section \ref{sec:EXP_DISC_SIMPLE}, we can evaluate the lower tail
probability of $d_{s}^{n}(\xv,\Xvbar)$ as follows:
    \begin{align}
        \sum_{\xvbar}P_{\Xv}(\xvbar)\openone\big\{d_{s}^{n}(\xv,\xvbar) \le nD\big\}
        &\le \sum_{\xvbar}P_{\Xv}(\xvbar)e^{t(nD-d_{s}^{n}(\xv,\xvbar))} \label{eq:EXP_GenStep2} \\
        &\,\,\dot{\le}\, \sum_{\xvbar}Q^{n}(\xvbar)e^{t(nD-d_{s}^{n}(\xv,\xvbar))}e^{\rbar(a(\xvbar)-n\phi_{a})} \label{eq:EXP_GenStep3} \\
        &= e^{n(tD-\rbar\phi_{a})}\prod_{i=1}^{n}\sum_{\xbar}Q(x)e^{\rbar a(\xbar)-td_{s}(x_{i},\xbar)},  \label{eq:EXP_GenStep4}
    \end{align}
where \eqref{eq:EXP_GenStep2} holds or any $t\ge0$ by upper bounding the indicator function,  
and \eqref{eq:EXP_GenStep3} holds for any $\rbar$ using \eqref{eq:EXP_CostProperty2} and 
\eqref{eq:EXP_CostProperty1}. From \eqref{eq:EXP_GenStep4}, we may set
    \begin{equation}
    R(D,\xv) = \sup_{t\ge0,\rbar} \rbar\phi_{a}-tD-\frac{1}{n}\sum_{i=1}^{n}\theta(x_{i},\rbar,t), \label{eq:EXP_RDx}
    \end{equation}
where 
    \begin{equation}
    \theta(x,\rbar,t) \defeq \log\EE_{Q}\big[e^{\rbar a(\Xbar)-td_{s}(x,\Xbar)}\big]. \label{eq:EXP_Theta}
    \end{equation}
Before proceeding, we present the following proposition.

\begin{prop} \label{prop:EXP_LimitQ}
     Consider the cost-constrained distribution $P_{\Xv}$ in \eqref{eq:EXP_Px_Multi},
     and assume that the input distribution $Q$ and auxiliary costs $\{a_{l}\}_{l=1}^{L}$ are
     such that assumptions of Proposition \ref{prop:EXP_SubExp} are satisfied.  
     For any function $f \,:\, \Xc \to \RR$, we have
        \begin{equation}
        \lim_{n\to\infty} \EE\bigg[\frac{1}{n}\sum_{i=1}^{n}f(X_i)\bigg] = \EE_{Q}[f(X)]
        \end{equation}
    provided that $\EE_{Q}[f(X)]$ exists.
\end{prop}
\begin{proof}
    See Appendix \ref{sec:EXP_LIMITQ_PROOF}.
\end{proof}

We can now simplify the exponent in \eqref{eq:EXP_FinalExpGen} as follows:
    \begin{align} \allowdisplaybreaks 
    & \EE\Big[\inf_{D\,:\,R(D,\Xv) \le R} D + R(D,\Xv)-R\Big] \\
    & \qquad = \EE\Big[\inf_{D}\sup_{\rho\ge1} D + \rho\big(R(D,\Xv)-R\big)\Big] \label{eq:EXP_GenStep5} \\
    & \qquad \ge \sup_{\rho\ge1} \EE\Big[\inf_{D} D + \rho\big(R(D,\Xv)-R\big)\Big] \\
    & \qquad = \sup_{\rho\ge1} \EE\Big[\inf_{D}\sup_{t\ge0,\rbar} D(1-\rho t) - \rho\Big(-\rbar\phi_{a}+\frac{1}{n}\sum_{i=1}^{n}\theta(X_{i},\rbar,t)+R\Big)\Big] \label{eq:EXP_GenStep7} \\
    & \qquad \ge \sup_{\rho\ge1}\sup_{\rbar} -\rho\Big(-\rbar\phi_{a}+ \EE\Big[\frac{1}{n}\sum_{i=1}^{n}\theta(X_{i},\rbar,1/\rho)\Big]+R\Big) \label{eq:EXP_GenStep10} \\
    & \qquad \to \sup_{\rho\ge1}\sup_{\rbar} \rho\Big(\rbar\phi_{a} - \EE_{Q}[\theta(X,\rbar,1/\rho)] - R\Big)\Big], \label{eq:EXP_GenStep11}
    \end{align}
where \eqref{eq:EXP_GenStep5} follows from \eqref{eq:EXP_Etilde3}, 
\eqref{eq:EXP_GenStep7} follows from \eqref{eq:EXP_RDx}, 
\eqref{eq:EXP_GenStep10} follows by replacing the supremum over $t\ge0$ by the choice
$t=1/\rho$, and \eqref{eq:EXP_GenStep11} follows from Proposition \ref{prop:EXP_LimitQ}.  

Substituting \eqref{eq:EXP_Theta} into \eqref{eq:EXP_GenStep11} setting 
$\rbar=\frac{1}{\rho}$, and taking the supremum over $a(\cdot)$, we 
obtain \eqref{eq:EXP_Ex_CC}, as desired.

\subsection{Proof of Proposition \ref{prop:EXP_LimitQ}} \label{sec:EXP_LIMITQ_PROOF}

We first present the proof in the case that there is $L=1$ auxiliary cost $a(\cdot)$ (with mean $\phi_{a}$) and 
no system cost constraint, and then discuss the changes required to handle the general case.
Throughout the proof, we define $a^{n}(\xv)\defeq\sum_{i=1}^{n}a(x_i)$ and 
$f^{n}(\xv)\defeq\sum_{i=1}^{n}f(x_i)$.  We use summations to denote averaging with respect to
$Q$, but the proof remains valid in the continuous case upon replacing these by integrals.

Let $\Xv$ be the random cost-constrained codeword, and define $\Xv' \sim Q^{n}(\xv')$.
From \eqref{eq:EXP_Px_Multi}, we have
    \begin{equation}
        \frac{1}{n}\EE\big[f^{n}(\Xv)\big] = \frac{1}{n}\frac{1}{\mu_{n}}\EE\Big[f^{n}(\Xv')\openone\big\{|a^{n}(\Xv')-n\phi_{a}|\le\delta\big\}\Big].
    \end{equation}
By a direct differentiation, this is equal to $\frac{d}{d\lambda}\big(\frac{1}{n}\log Z(\lambda)\big)$
evaluated at $\lambda=0$, where
    \begin{equation}
        Z(\lambda) \defeq \EE\Big[e^{\lambda f^{n}(\Xv')}\openone\big\{|a^{n}(\Xv')-n\phi_{a}|\le\delta\big\}\Big].
    \end{equation}
Expanding the expectation and using the inverse Laplace transform relation
    \begin{equation}
    \openone\{z \ge 0\} = \frac{1}{2\pi j}\int_{u-j\infty}^{u+j\infty}\frac{e^{tz}}{t}dt
    \end{equation} 
for $u>0$, we have the following:
    \begin{align}
        Z(\lambda) &= \sum_{\xv'}Q^{n}(\xv')e^{\lambda f^{n}(\xv')}\Big(\openone\{a^{n}(\xv') \le n\phi_{a}+\delta\}-\openone\{a^{n}(\xv') \le n\phi_{a}-\delta\}\Big) \\
                   &= \frac{1}{2\pi j}\sum_{\xv'}Q^{n}(\xv')e^{\lambda f^{n}(\xv')}\int_{u-j\infty}^{u+j\infty}e^{t(n\phi_{a}-a^{n}(\xv'))}\frac{e^{t\delta}-e^{-t\delta}}{t}dt \\
                   &= \frac{1}{2\pi j}\int_{u-j\infty}^{u+j\infty}\frac{e^{t\delta}-e^{-t\delta}}{t}e^{n\phi_{a}t}\bigg(\sum_{x'}Q(x')e^{-ta(x')+\lambda f(x')}\bigg)^{n}dt.
    \end{align}
Denoting the derivative of $Z(\cdot)$ by $Z'(\cdot)$, we have
    \begin{align}
        Z'(0) &= \frac{n}{2\pi j}\int_{u-j\infty}^{u+j\infty} \frac{e^{t\delta}-e^{-t\delta}}{t}e^{n\phi_{a}t}\bigg(\sum_{x'}Q(x')e^{-ta(x')}\bigg)^{n-1}\sum_{x'}Q(x')f(x')e^{-ta(x')}dt \\
              &= \frac{n}{2\pi j}\int_{u-j\infty}^{u+j\infty} \frac{e^{t\delta}-e^{-t\delta}}{t}e^{n\phi_{a}t}\bigg(\sum_{x'}Q(x')e^{-ta(x')}\bigg)^{n}\frac{\sum_{x'}Q(x')f(x')e^{-ta(x')}}{\sum_{x'}Q(x')e^{-ta(x')}}dt.
    \end{align}
Finally, using the assumption that $\EE_{Q}[a(X)^2]<\infty$ and applying the saddlepoint
method \cite[Ch. 4-5]{deBruijn} (see also \cite[Sec. 4.2-4.3]{MerhavPhysics}), we obtain
    \begin{equation}
        \frac{d}{d\lambda}\Big(\frac{1}{n}\log Z(\lambda)\Big)\Big|_{\lambda=0} = \frac{Z'(0)}{Z(0)} \to \frac{\sum_{x'}Q(x')f(x')e^{-t_{0}a(x')}}{\sum_{x'}Q(x')e^{-t_{0}a(x')}}, \label{eq:EXP_LimitQ_Final}
    \end{equation}
where $t_{0}$ is the zero of the derivative (saddlepoint) of the function $h(t)=\phi_{a}t + \log\EE_{Q}[e^{-ta(X)}]$. 
Since $\phi_{a}=\EE_{Q}[a(X)]$ by definition, it is easily verified that $t_{0}=0$, and thus the right-hand
side of \eqref{eq:EXP_LimitQ_Final} equals $\EE_{Q}[f(X)]$, as desired.

In the case of multiple auxiliary costs, the argument is similar, but with $ta(\cdot)$ 
replaced by $\sum_{l}t_{l}a_{l}(\cdot)$. The system cost $c(x)$ in \eqref{eq:EXP_SetDn} can be
handled similarly provided that $\EE_{Q}[c(X)]\le\Gamma$, which is an assumption of the proposition.

\bibliographystyle{IEEEtran}
\bibliography{18-MultiUser,18-SingleUser,35-Other,12-Paper} 

\end{document}

%% file: preamble.tex
\newcommand{\openone}{\leavevmode\hbox{\small1\normalsize\kern-.33em1}} 
\newcommand{\defeq}{\triangleq}
\newcommand{\rcux}{\mathrm{rcux}}
\newcommand{\var}{\mathrm{Var}} 
\long\def\symbolfootnote[#1]#2{\begingroup\def\thefootnote{\fnsymbol{footnote}}\footnote[#1]{#2}\endgroup}

\newcommand{\Eex}{E_{\mathrm{ex}}}
\newcommand{\Eexiid}{E_{\mathrm{ex}}^{\mathrm{iid}}}

\newcommand{\Eexcc}{E_{\mathrm{ex}}^{\mathrm{cc}}}
\newcommand{\Eexhatcc}{\hat{E}_{\mathrm{ex}}^{\mathrm{cc}}}
\newcommand{\Eexcost}{E_{\mathrm{ex}}^{\mathrm{cost}}}

\newcommand{\Exiid}{E_{\mathrm{x}}^{\mathrm{iid}}}
\newcommand{\Excc}{E_{\mathrm{x}}^{\mathrm{cc}}}
\newcommand{\Excost}{E_{\mathrm{x}}^{\mathrm{cost}}}
\newcommand{\Excoststar}{E_{\mathrm{x}}^{\mathrm{cost}^{*}}}

\newcommand{\SetScc}{\mathcal{S}}
\newcommand{\SetSncc}{\mathcal{S}_{n}}
\newcommand{\SetTcc}{\mathcal{T}}

\newcommand{\SetTncc}{\mathcal{T}_{n}}
\newcommand{\PXv}{P_{\Xv}}

\newcommand{\rbar}{\overline{r}}

\newcommand{\Ptilde}{\widetilde{P}}

\newcommand{\mbar}{\overline{m}}

\newcommand{\xbar}{\overline{x}}
\newcommand{\Xbar}{\overline{X}}
\newcommand{\xv}{\boldsymbol{x}}
\newcommand{\xvbar}{{\overline{\boldsymbol{x}}}}
\newcommand{\yv}{\boldsymbol{y}}
\newcommand{\Xv}{\boldsymbol{X}}
\newcommand{\Xvbar}{\overline{\boldsymbol{X}}}
\newcommand{\Yv}{\boldsymbol{Y}}

\newcommand{\Ac}{\mathcal{A}}
\newcommand{\Cc}{\mathcal{C}}
\newcommand{\Dc}{\mathcal{D}}
\newcommand{\Ec}{\mathcal{E}}
\newcommand{\Fc}{\mathcal{F}}
\newcommand{\Pc}{\mathcal{P}}
\newcommand{\Xc}{\mathcal{X}}
\newcommand{\Yc}{\mathcal{Y}}

\newcommand{\EE}{\mathbb{E}}
\newcommand{\PP}{\mathbb{P}}
\newcommand{\RR}{\mathbb{R}}

\newcommand{\Csf}{\mathsf{C}}

\newcommand{\Lsf}{\mathsf{L}}

\newcommand{\al}{\{a_l\}}

\newcommand{\brl}{\{\overline{r}_l\}}

%% file: ExpurgatedPaperFull.bbl
\begin{thebibliography}{10}
\providecommand{\url}[1]{#1}
\csname url@samestyle\endcsname
\providecommand{\newblock}{\relax}
\providecommand{\bibinfo}[2]{#2}
\providecommand{\BIBentrySTDinterwordspacing}{\spaceskip=0pt\relax}
\providecommand{\BIBentryALTinterwordstretchfactor}{4}
\providecommand{\BIBentryALTinterwordspacing}{\spaceskip=\fontdimen2\font plus
\BIBentryALTinterwordstretchfactor\fontdimen3\font minus
  \fontdimen4\font\relax}
\providecommand{\BIBforeignlanguage}[2]{{%
\expandafter\ifx\csname l@#1\endcsname\relax
\typeout{** WARNING: IEEEtran.bst: No hyphenation pattern has been}%
\typeout{** loaded for the language `#1'. Using the pattern for}%
\typeout{** the default language instead.}%
\else
\language=\csname l@#1\endcsname
\fi
#2}}
\providecommand{\BIBdecl}{\relax}
\BIBdecl

\bibitem{Shannon}
C.~E. Shannon, ``A mathematical theory of communication,'' \emph{Bell Syst.
  Tech. Journal}, vol.~27, pp. 379--423, July and Oct. 1948.

\bibitem{Gallager}
R.~Gallager, \emph{Information Theory and Reliable Communication}.\hskip 1em
  plus 0.5em minus 0.4em\relax John Wiley \& Sons, 1968.

\bibitem{Finite}
Y.~Polyanskiy, V.~Poor, and S.~Verd\'{u}, ``Channel coding rate in the finite
  blocklength regime,'' \emph{IEEE Trans. Inf. Theory}, vol.~56, no.~5, pp.
  2307--2359, May 2010.

\bibitem{ExpurgCKM}
I.~Csisz{\'a}r, J.~K\"{o}rner, and K.~Marton, ``A new look at the error
  exponent of discrete memoryless channels,'' in \emph{IEEE Int. Symp. Inf.
  Theory}, Ithaca, NY, 1977.

\bibitem{CsiszarBook}
I.~Csisz\'{a}r and J.~K\"{o}rner, \emph{Information Theory: Coding Theorems for
  Discrete Memoryless Systems}, 2nd~ed.\hskip 1em plus 0.5em minus 0.4em\relax
  Cambridge University Press, 2011.

\bibitem{Csiszar1}
------, ``Graph decomposition: A new key to coding theorems,'' \emph{IEEE
  Trans. Inf. Theory}, vol.~27, no.~1, pp. 5--12, Jan. 1981.

\bibitem{MerhavErasure}
N.~Merhav, ``Error exponents of erasure/list decoding revisited via moments of
  distance enumerators,'' \emph{IEEE Trans. Inf. Theory}, vol.~54, no.~10, pp.
  4439--4447, Oct. 2008.

\bibitem{MerhavIC}
R.~Etkin, N.~Merhav, and E.~Ordentlich, ``Error exponents of optimum decoding
  for the interference channel,'' \emph{IEEE Trans. Inf. Theory}, vol.~56,
  no.~1, pp. 40--56, 2010.

\bibitem{MerhavPhysics}
N.~Merhav, ``Statistical physics and information theory,'' \emph{Foundations
  and Trends in Comms. and Inf. Theory}, vol.~6, no. 1-2, pp. 1--212, 2009.

\bibitem{Merhav}
N.~Merhav, G.~Kaplan, A.~Lapidoth, and S.~Shamai, ``On information rates for
  mismatched decoders,'' \emph{IEEE Trans. Inf. Theory}, vol.~40, no.~6, pp.
  1953--1967, Nov. 1994.

\bibitem{Csiszar2}
I.~Csisz\'{a}r and P.~Narayan, ``Channel capacity for a given decoding
  metric,'' \emph{IEEE Trans. Inf. Theory}, vol.~45, no.~1, pp. 35--43, Jan.
  1995.

\bibitem{MMRevisited}
A.~Ganti, A.~Lapidoth, and E.~Telatar, ``Mismatched decoding revisited:
  {G}eneral alphabets, channels with memory, and the wide-band limit,''
  \emph{IEEE Trans. Inf. Theory}, vol.~46, no.~7, pp. 2315--2328, Nov. 2000.

\bibitem{JournalSU}
J.~Scarlett, A.~Martinez, and A.~{Guill{\'e}n i F\`{a}bregas}, ``Mismatched
  decoding: Error exponents, second-order rates and saddlepoint
  approximations,'' \emph{IEEE Trans. Inf. Theory}, vol.~60, no.~5, pp.
  2647--2666, May 2014.

\bibitem{ExpurgJelenik}
F.~Jelinek, ``Evaluation of expurgated bound exponents,'' \emph{IEEE Trans.
  Inf. Theory}, vol.~14, no.~3, pp. 501--505, 1968.

\bibitem{ExpurgBlahut}
R.~Blahut, ``Composition bounds for channel block codes,'' \emph{IEEE Trans.
  Inf. Theory}, vol.~23, no.~6, pp. 656--674, 1977.

\bibitem{ExpurgOmura}
J.~K. Omura, ``Expurgated bounds, {B}hattacharyya distance, and rate distortion
  functions,'' \emph{Inf. and Control}, vol.~24, no.~4, pp. 358 -- 383, 1974.

\bibitem{JSCC4}
I.~Csisz{\'a}r, ``On the error exponent of source-channel transmission with a
  distortion threshold,'' \emph{IEEE Trans. Inf. Theory}, vol.~28, no.~6, pp.
  823--828, Nov. 1982.

\bibitem{Convex}
S.~Boyd and L.~Vandenberghe, \emph{Convex Optimization}.\hskip 1em plus 0.5em
  minus 0.4em\relax Cambridge University Press, 2004.

\bibitem{RefinementJournal}
Y.~Altu\u{g} and A.~B. Wagner, ``Refinement of the random coding bound,'' 2014,
  http://arxiv.org/abs/1312.6875.

\bibitem{PaperRefinement}
J.~Scarlett, A.~Martinez, and A.~{Guill{\'e}n i F\`{a}bregas}, ``A derivation
  of the asymptotic random-coding prefactor,'' in \emph{Allerton Conf. on
  Comm., Control and Comp.}, Monticello, IL, 2013.

\bibitem{PaperITA}
J.~Scarlett, A.~Martinez, and A.~{Guill\'{e}n i F\`{a}bregas},
  ``Cost-constrained random coding and applications,'' in \emph{Inf. Theory and
  Apps. Workshop}, San Diego, CA, Feb. 2013.

\bibitem{Cover}
T.~M. Cover and J.~A. Thomas, \emph{Elements of Information Theory}.\hskip 1em
  plus 0.5em minus 0.4em\relax John Wiley \& Sons, Inc., 2001.

\bibitem{ShannonZero}
C.~E. Shannon, ``The zero error capacity of a noisy channel,'' \emph{IRE Trans.
  Inf. Theory}, vol.~2, no.~3, pp. 8--19, Sept. 1956.

\bibitem{DyachkovCC}
A.~G. D'yachkov, ``Bounds on the average error probability for a code ensemble
  with fixed composition,'' \emph{Prob. Inf. Transm.}, vol.~16, no.~4, pp.
  3--8, 1980.

\bibitem{Minimax}
K.~Fan, ``Minimax theorems,'' \emph{Proc. Nat. Acad. Sci.}, vol.~39, pp.
  42--47, 1953.

\bibitem{PaperIZS}
J.~Scarlett, A.~Martinez, and A.~{Guill\'en i F\`abregas}, ``Expurgated
  random-coding ensembles: Exponents, refinements and connections,'' in
  \emph{Int. Zurich Sem. on Comms.}, Feb. 2014.

\bibitem{PhysicsMezard}
M.~M\'ezard and A.~Montanari, \emph{Information, Physics and
  Computation}.\hskip 1em plus 0.5em minus 0.4em\relax Oxford University Press,
  2009.

\bibitem{Derrida80a}
B.~Derrida, ``Random-energy model: Limit of a family of disordered models,''
  \emph{Phys. Rev. Lett.}, vol.~45, no.~2, pp. 79--82, 1980.

\bibitem{Derrida80b}
------, ``The random energy model,'' \emph{Physics Reports}, vol.~67, no.~1,
  pp. 29--35, 1980.

\bibitem{Derrida81}
------, ``Random-energy model: An exactly solvable model for disordered
  systems,'' \emph{Phys. Rev. Lett.}, vol.~24, no.~5, pp. 2613--2626, 1981.

\bibitem{TwoChannels}
P.~Elias, ``Coding for two noisy channels,'' in \emph{Third London Symp. Inf.
  Theory}, 1955.

\bibitem{Dobrushin}
R.~L. Dobrushin, ``Asymptotic estimates of the probability of error for
  transmission of messages over a discrete memoryless communication channel
  with a symmetric transition probability matrix,'' \emph{Theory Prob. Apps.},
  vol.~7, no.~3, pp. 270--300, 1962.

\bibitem{RefinementSP}
Y.~Altu\u{g} and A.~B. Wagner, ``Refinement of the sphere-packing bound:
  Asymmetric channels,'' \emph{IEEE Trans. Inf. Theory}, vol.~60, no.~3, pp.
  1592--1614, March 2013.

\bibitem{FiniteThesis}
Y.~Polyanskiy, ``Channel coding: Non-asymptotic fundamental limits,'' Ph.D.
  dissertation, Princeton University, 2010.

\bibitem{deBruijn}
N.~G. {de Bruijn}, \emph{Asymptotic Methods in Analysis}.\hskip 1em plus 0.5em
  minus 0.4em\relax Dover Publications, 1981.

\end{thebibliography}
